\PassOptionsToPackage{dvipsnames}{xcolor}

\documentclass[a4paper,twocolumn,unpublished,11pt]{quantumarticle}
\pdfoutput=1

\usepackage[T1]{fontenc}

\usepackage{graphicx}
\usepackage{amsmath}
\usepackage{amssymb}
\usepackage{amsthm}
\usepackage{mathtools}
\usepackage{physics}
\usepackage{dsfont}
\usepackage{bbm}
\usepackage{braket}
\usepackage{longtable}
\usepackage{enumitem}
\usepackage{float}

\usepackage[super,sort&compress]{natbib}

\usepackage{xcolor}

\usepackage{tikz}
\usetikzlibrary{
    positioning,
    arrows.meta,
    shapes.geometric,
    fit,
    calc
}

\usepackage{tcolorbox}
\tcbuselibrary{breakable,skins}

\definecolor{LightGray}{RGB}{220,220,220}
\definecolor{myred2}{RGB}{255,19,0}
\definecolor{myred}{RGB}{15,122,98}
\definecolor{myblue}{RGB}{14,81,167}
\definecolor{myorange}{RGB}{255,129,0}
\definecolor{mygreen}{RGB}{0,146,44}

\newtheorem{thm}{Theorem}[section]
\newtheorem*{thm*}{Theorem}
\newtheorem{lem}[thm]{Lemma}
\newtheorem{cor}[thm]{Corollary}

\newtheorem{prop}[thm]{Proposition}

\newcommand{\cB}{\mathcal{B}}

\newcommand{\cF}{\mathcal{F}}
\newcommand{\cG}{\mathcal{G}}
\newcommand{\cH}{\mathcal{H}}

\newcommand{\cJ}{\mathcal{J}}
\newcommand{\cK}{\mathcal{K}}

\newcommand{\cO}{\mathcal{O}}

\newcommand{\cS}{\mathcal{S}}

\newcommand{\cZ}{\mathcal{Z}}

\def\tr{\operatorname{tr}}%
\def\vun#1{\left(#1\right)}%

\usepackage[normalem]{ulem}

\DeclareMathOperator{\supp}{supp}

\usepackage[pagebackref, colorlinks = true, linkcolor = myred, urlcolor  = myred, citecolor = myred]{hyperref}

\begin{document}

\title{\vspace{-0.2cm}  Efficient Computation of QKD Key Rates without Semidefinite Programming}

\author{Bence Temesi}
%\email{bence.temesi@itp.uni-hannover.de}
\affiliation{Institute for Theoretical Physics, Leibniz Universit\"at Hannover, Hannover, Germany}
\author{Antoine  Gansel}
%\email{enter-email@thispoi.nt}
\affiliation{Chair of Data Security and Cryptography, University of Regensburg, Regensburg, Germany}
\author{Gereon Ko\ss mann}
%\email{rene.schwonnek@itp.uni-hannover.de}
\affiliation{Institute for Quantum Information, RWTH Aachen University, Aachen, Germany}
\author{Ren\'e Schwonnek}
%\email{rene.schwonnek@itp.uni-hannover.de}
\affiliation{Institute for Theoretical Physics, Leibniz Universit\"at Hannover, Hannover, Germany}

\maketitle

\fontfamily{lmr}\selectfont

\begin{abstract}
Translating observed data into a reliable estimate of the secure key rate is a crucial step for operating a quantum key distribution device. We provide a computational method for this task that only requires eigenvalue computations and is therefore both fast and resource efficient. In contrast, existing approaches rely on semidefinite programming or programming on the entropy cone, whose memory requirements can scale as $d^4$ in the underlying Hilbert-space dimension. %depending on how one incorporates the constraints.
Our method reduces this requirement to $d^2$. A minimal implementation of our algorithm takes fewer than 100 lines of Common Lisp. We demonstrate real-time key-rate estimation on a Raspberry Pi with a 1 GB memory and a Cortex-A53 processor. Despite these modest resources, our implementation outperforms existing workstation-based benchmarks by several orders of magnitude. Non-numerical verification can be incorporated with little overhead using rational approximations. These results open the way toward embedding complete numerical security analysis directly into qkd hardware.
\end{abstract}
\vspace{-0.5cm}
\section{Introduction}\label{sec:intro}
Over the past decades, quantum cryptography has developed from the theoretical proposal of Bennett and Brassard in 1984 \cite{Bennett2014} into an increasingly mature technology. Quantum key distribution systems\cite{Gisin2002,Lo2014,Diamanti2016,Pirandola2020} and also quantum random number generators\cite{Ma2016}  are now actually approaching a wider range of commercial deployment. Furthermore, academic research holds many demonstrator systems\cite{Boaron2018,Minder2019,Wei2020,Pittaluga2021,Zahidy2024,Li2026} and proof-of-principle experiments \cite{Liu2022,Zhang2022,Nadlinger2022,Primaatmaja2023} that will continue to broaden the range of cryptographic quantum technologies and may become practically relevant in a not to far future.

Turning these advances \cite{Pirandola2020} into reliable products requires however more than a secure physical design. A formal security proof\cite{Scarani2009} based on fundamental principles of quantum mechanics is an equally important part of a fully functioning system.
In the last years, theoretical research has 
taken essential steps by providing functional frameworks for performing a full security analysis for real devices.  In particular, modern proof systems are commonly based on the left over hash lemma \cite{Renner2005,Tomamichel2011,Dupuis2023} combined with an entropy accumulation theorem \cite{ArnonFriedman2018,Dupuis2020,Metger2024,Arqand2025} or closely related techniques \cite{MarginalConstraintEAT,PetzRenyiBounds}. They allow to establish $\varepsilon$-composable \cite{RennerWolf2005,RENNER2008} security against general quantum attacks. The input to these frameworks is a bound on the extractable randomness, formulated by the conditional entropy $H(X|E)$\cite{Wilde2016} that an attacker $E$ has on the raw data $X$ from which a device will produce a key. In essence, this number determines the compression rate of a privacy amplification step, i.e. the amount of information that has to be sacrificed in order to guarantee that Eve has almost no knowledge on what remains. 

A central task for a security guarantee is hence to translate experimental data, which is for example obtained in the parameter estimation step of a protocol or already within a calibration phase, into a rigorous and sufficiently tight estimate on the conditional von Neumann entropy. Besides in highly symmetric protocols, like  BB84, this task, that deals with noisy real world data, demands to use numerical methods. Performing this task on embedded hardware usable in a demonstrator for a full CV-qkd system\cite{Laudenbach2018,Kanitschar2023,Navarro2026} was the original challenge motivating the research presented in this work. 

\section{Results}\label{sec:results}
\vspace{-0.3cm}Several numerical approaches have been developed to estimate the conditional entropy from observed measurement data. While these methods have made the numerical treatment of increasingly general qkd protocols possible, their computational demands can become substantial already for moderately sized instances. In particular, existing approaches rely on semidefinite optimization\cite{Coles2016,Winick2018,Tan2021,Hu2022,Arajo2023,Kossmann2026} or optimization on a problem-specific cone\cite{Fawzi2023,He2025,Lorente2025,qics}. Both memory requirements and runtime grow rapidly with the Hilbert-space dimension $d$ and the number of constraints $n$. Reported computations can therefore require substantial workstation or HPC resources and, for some instances, runtime of minutes to hours. 

For practical qkd systems operating under changing conditions, however, the secure key rate should ideally be recomputed repeatedly from live data, potentially directly on embedded hardware. The computational cost of existing methods makes such on-device evaluation difficult. Existing qkd systems therefore, so far have to resort to precomputed lookup tables or analytically tractable approximations\cite{Laudenbach2018,Denys2021}, sacrificing either adaptability to changing conditions or tightness of the resulting bound. Efficient and certifiable key-rate estimation is thus a part of the technological challenge itself.

In this work, we introduce a compact and resource-efficient alternative that avoids semidefinite programming altogether. We bypass costly numerical optimizations over the full quantum state space by a candidate-and-certificate method that iteratively exploits the optimality of Gibbs states for entropy optimization. The resulting algorithm requires only standard matrix diagonalizations and smallest-eigenvalue computations%, with memory requirements scaling as $d^2$ and runtime requirements $d^3$
%$d_A^2 d_B^2$ and runtime requirements scaling as 
%$d_A^3 d_B^3$ 
per iteration. Indeed, we do not need any discretization or numerical approximation of the conditional von Neumann entropy. This makes the complete security calculation sufficiently lightweight to be performed directly on embedded hardware. We demonstrate this by implementing a solver on a Raspberry Pi 3. It has only $1\mathrm{GiB}$ of memory and a Cortex-A53 processor, which is a very common CPU for embedded devices including home routers, many SoCs, and the fire tv stick. Despite these modest resources, it  outperforms existing workstation-based benchmarks by orders of magnitude.
\begin{figure}
    (a) Performance Benchmark on a Raspberry Pi 3\\
                \includegraphics[width=0.99\linewidth]{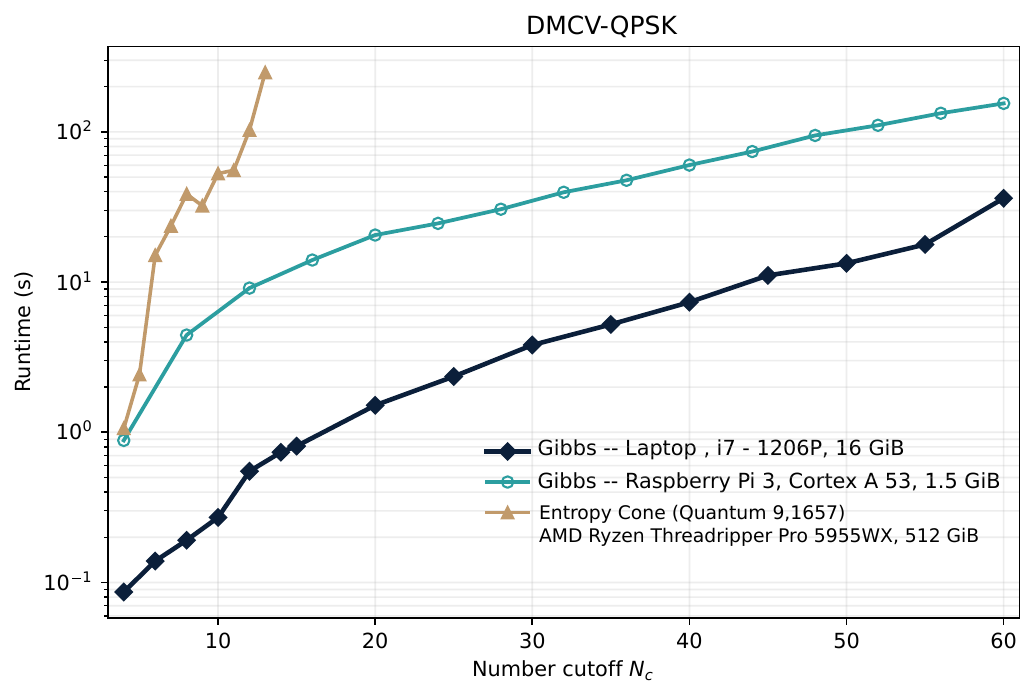}\\[0.1cm]
    (b) Setup and Entropyflow
                \begin{center}
                    \includegraphics[width=0.6\linewidth]{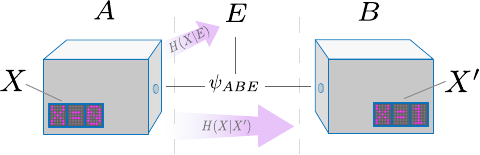}\\[0.1cm]
                    \includegraphics[width=0.7\linewidth]{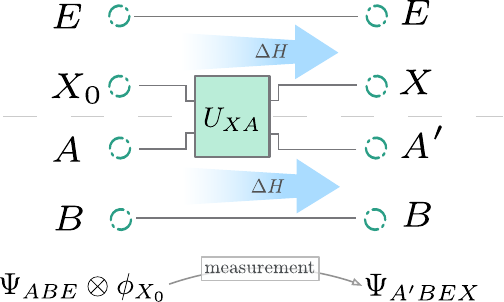}\\[0.1cm]
                \end{center} 
    (c) Inner Candidate and Outer certificate 
                \includegraphics[width=0.99\linewidth]{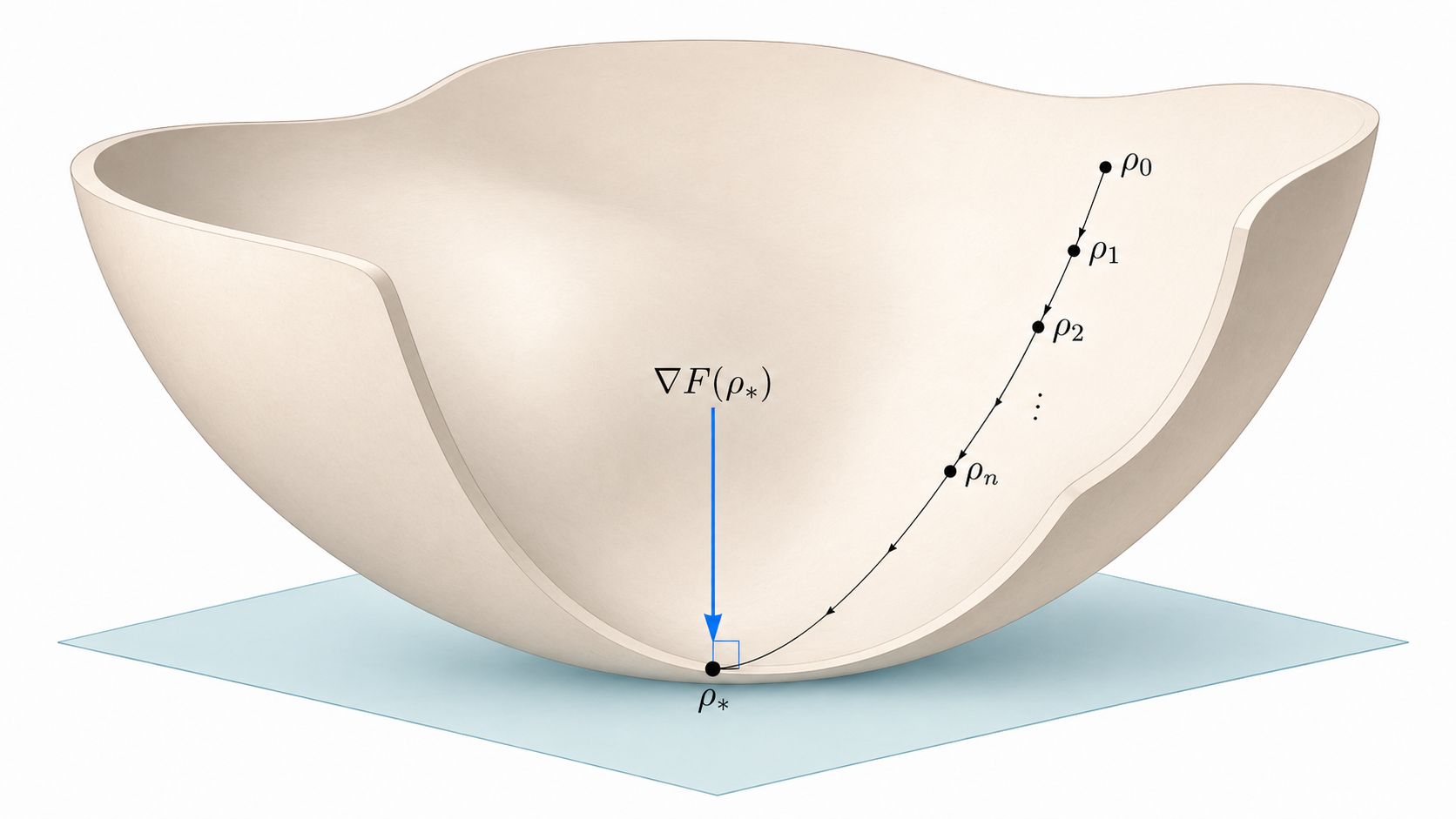}

    \caption{
    (a) Runtime of our algorithm on a laptop and on a raspberry 3 in comparison to existing methods that run on workstations\cite{Lorente2025}. Benchmark for QPSK-DMCV-QKD  based on example data from \cite{Lorente2025}. Details in Sec.\ref{sec:bench}.
    (b) Basic setting of an entanglement based Alice-Bob-Eve protocol in stylised boxes and as circuit diagram. The uncertainty $H(X|E)$ equals the entropy production $\Delta H$ of the measurement process $X_0E\rightarrow XE$. When dilated\cite{Wilde2016} to a unitary $U_{AX}$, the key map $\Phi:A\rightarrow A'$ is the complementary channel to this process. Its entropy production is also $\Delta H$. (c) An AI's perceptions of the basic geometric intuition of our algorithm.  We first produce a inner candidate $\rho^*$. Since the optimization is convex we can lower bound the graph  $(\rho,F(\rho))$ of the target function by a tangent hyperplane at a point $(\rho^*,F(\rho^*))$. This relaxations gets tight at the optimum. }
    \label{fig:bild1}
\end{figure}

\subsection{Setting and problem statement}\label{sec:setting}
We consider a standard quantum key distribution scenario\cite{Pirandola2020,Diamanti2016} in the entanglement based picture, see Fig.~\ref{fig:bild1}. Prepare and measure protocols can be brought into this picture by the source replacement method \cite{Curty2004,Tupkary2025Review,vanLuijk2024}.
The scenario involves two legitimate parties, Alice and Bob who aim to establish a shared secret key from a quantum state $\rho_{AB}$. In a worst case, an attacker Eve has full access to the source that produces this state and can store a purification $\psi_{ABE}$ of it in a quantum memory. 
The information Alice and Bob have on the protocol and on the potentially compromised state $\rho_{AB}$ is modeled by linear constraints $\tr\vun{\rho_{AB}M_i}=m_i$ for bipartite measurements $\{M_i\}$ on Alice's and Bob's system and observed statistics $(m_i)$ for $i=1,\ldots, n$.
They contain all measurement data obtained during the calibration and parameter estimation phase but can also express symmetries and other modeling assumptions on a concrete qkd device. Linearity of constraints follows from the stochastic structure of quantum mechanics. For $d \equiv d_Ad_B$ let
\begin{align}\label{eq:feasible-set}
    \mathcal{F}=\{\rho \ | \ \tr\vun{\rho M_i}=m_i, \ \rho\geq 0, \ \tr \rho = 1\}    
\end{align}
be the set of feasible states. We focus on the standard case of protocols with one-way error correction. Here one party declares their dataset to be the target the other party has to recover. For the moment let Alice be the party with the correct data. For  CV-qkd protocols\cite{Laudenbach2018} Bob usually takes this role, which in our computations only results in a relabeling of variables.
The measurement of this raw data, i.e. the raw key, is modeled by a quantum channel $\Phi$ \cite{Coles2016, Tan2021} that maps Alice's initial system to the state after the raw key data has been produced. The quantity we want to compute in this work is the minimal conditional entropy
\begin{align}\label{eq:qkdproblem}
F^* \coloneqq \inf_{\psi_{ABE}}  H(X|E) \ \text{ s.th. } \rho_{AB}\in\mathcal{F}.
\end{align}
Operationally the value of \eqref{eq:qkdproblem} can be understood as worst case bound on the uncertainty\cite{Shannon2001} Eve has in a single round on the raw data $X$. 
In practice it determines the securely extractable randomness\cite{Renner2005}. 
Reliable bounds on it are an essential ingredient of the security assessment of a qkd device. On one hand they allow for the direct computation of the asymptotic keyrate via the Devetak Winter formula \cite{Devetak2005} and on the other they are the relevant single shot quantities from which min-tradeoff functions in the entropy accumulation framework\cite{Metger2024} for finite size analyses are constructed. 

By using basic duality arguments\cite{Tomamichel2016,Tan2021}, the optimization in \eqref{eq:qkdproblem} can be reduced to an optimization that acts only on the Alice-Bob system. Here we have to solve the equivalent problem
\begin{align}\label{eq:qkdAB}
    F^*=\inf_{\rho \in\mathcal{F} } D\vun{\rho\Vert \Phi(\rho)}
\end{align}
in which we minimize over the Umegaki relative entropy $D(\rho\Vert \sigma) = \tr\rho\vun{\log \rho - \log \sigma}$. This reformulation is not merely a useful mathematical tool, it also reveals a fundamental intuition on why qkd works where classical methods do not. The primordial problem \eqref{eq:qkdproblem} contains the system of Eve. It is inaccessible to legitimate users of a cryptographic device. The duality step, illustrated in Fig.~\ref{fig:bild1}(a), hinges on the concept of purifications, which is a genuine feature of quantum theories \cite{Plvala2023}. The resulting expression \eqref{eq:qkdAB}, which can be understood as minimal entropy production of the measurement process \cite{Tan2021}, depends the quantum systems located in Alice's and Bobs labs.
So in principle, they can gather the information needed for determining $F^*$ through local actions and classical communication. Finding those bounds from incomplete knowledge is exactly the problem our method solves. 

\subsection{A Candidate and Certificate method}
There are three central challenges in the computation of $F^*$. First, the optimization is performed over the full state space of the underlying quantum system. For a Hilbert-space dimension $d$, a density operator already comprises $d^2-1$ independent real parameters, subject in addition to positivity and the constraints imposed by $\mathcal{F}$. Second, the objective itself is nonlinear, but convex in the quantum state. The quantum relative entropy involves matrix logarithms and, in general, does not reduce to a standard linear or quadratic optimization problem. Finally, a numerical approximation to the optimum is not sufficient by itself. For cryptographic applications, the resulting value must be accompanied by a rigorous bound that remains valid despite finite numerical precision. Our method addresses these three difficulties by separating the computation into a fast \emph{candidate} finding phase and a subsequent \emph{certification} phase. A sketch of the underlying geometry is illustrated in Fig.~\ref{fig:bild1}(b). A flow diagram of our algorithm is depicted in Fig.~\ref{fig:bild2}. A detailed derivation is given in Sec.~\ref{methods} and in the appendix. 
\begin{figure}[t]
    \centering
    \includegraphics[width=0.8\linewidth]{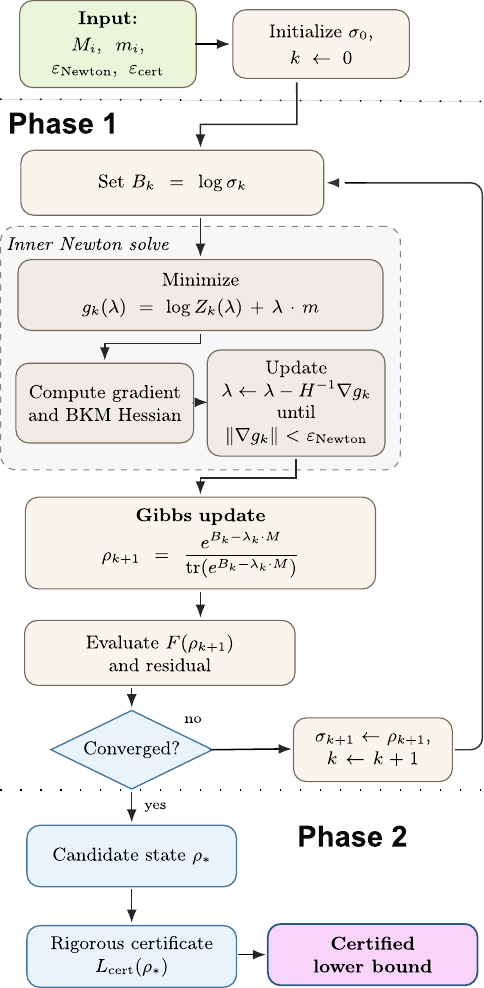}
    \caption{Flow diagram of the candidate-and-certificate algorithm. Phase 1 generates with an iterative optimization over adopted Gibbs states a candidate state, which is in Phase 2 converted into a lower bound. The core computation effort is the loop for a newton method.}
    \label{fig:bild2}
\end{figure}

\textbf{Phase 1: Candidate.}
Rather than optimizing directly over all matrix elements of the quantum state at once, we generate a sequence of feasible Gibbs states that converges to the optimum. The construction is most naturally formulated in analogy with equilibrium thermodynamics. Starting from a full-rank state $\rho_0$, each iteration determines a set of Lagrange multipliers $\lambda=(\lambda_1,\ldots,\lambda_n)$ associated with the $n$ constraints that form $\mathcal{F}$. 

For fixed $\rho_k$, we introduce the effective Hamiltonian
\begin{align}
H_k(\lambda) = \Phi(\rho_k)-\lambda\cdot M
\end{align}
and its partition function
\begin{align}
Z_k(\lambda) = \operatorname{tr}\exp\vun{H_k(\lambda)}.
\end{align}
The next candidate state is the corresponding Gibbs state,
\begin{align}\label{eq:gibbs1}
\rho_{k+1}(\lambda) = \frac{\exp[H_k(\lambda)]}{Z_k(\lambda)}.
\end{align}
with  multipliers $\lambda$ that are obtained by minimizing the convex log-partition functional
\begin{align}
g_k(\lambda) = \log Z_k(\lambda) +\lambda\cdot m.
\end{align}
Indeed, its stationarity conditions are precisely the linear constraints,
\begin{align}
\frac{\partial g_k}{\partial\lambda_i} = m_i-\operatorname{tr}\left(\rho_{k+1} (\lambda)M_i\right) =0,
\end{align}
which ensures $\rho_{k+1}(\lambda)\in\mathcal{F}$. The central simplification is therefore that the original matrix-valued optimization is replaced, at every iteration, by a smooth convex optimization of $g_k$ over the real parameters $\lambda$. Repeating this construction produces a sequence of feasible states approaching a candidate $\rho_{\mathrm{cand}}$ and a candidate value $F_{\mathrm{cand}}$, see Fig.~\ref{fig:bild1}(c).

The optimization over the multipliers $\lambda$ is convex and can be performed efficiently by Newton-type methods\cite{Boyd2004}. We test and finetune several variations. Both the gradient and Hessian of $g_k$ can be evaluated analytically; the latter takes the form of a Gibbs susceptibility, or equivalently a Bogoliubov--Kubo--Mori covariance matrix\cite{Petz1993}. Computationally, the dominant matrix operation is the diagonalization of the current effective Hamiltonian $H_k(\lambda)$, whose eigensystem provides the Gibbs state as well as the quantities required for the newton step. 

\textbf{Phase 2: Certificate.}
A feasible candidate alone provides only one side of the optimum. To obtain a rigorous lower bound, we exploit the convexity of the original optimization problem \eqref{eq:qkdAB}. Let $F(\rho)$ be the relative entropy on the r.h.s of  \eqref{eq:qkdAB}. By convexity of $F$ in its argument, any feasible candidate $\sigma$ and a gradient $G_{\sigma}$ of $F$ at $\sigma$  gives a supporting hyperplane to the Graph $(\sigma,F)$, depicted as blue plane in  Fig.~\ref{fig:bild1}(c). For every $\sigma$, this hyperplane gives the affine lower bound 
\begin{align}
F(\rho) \geq F(\sigma) + \operatorname{tr}\vun{G_{\sigma}(\rho-\sigma)}.
\end{align}
Let $\rho^*$ denote the optimizer of \eqref{eq:qkdAB}. Setting $\rho=\rho^*$ and following the steps in Sec.~\ref{sec:outer-certificate} leads to the certificate
\begin{align}\label{eq:cert1}
F^* \geq \lambda_{\min}\left(G_\sigma+\lambda\cdot M\right)-\lambda\cdot m \equiv F_{\mathrm{cert}} .
\end{align}
Here $\lambda_{min}$ refers to the smallest eigenvalue. %\GK{of what, needs some explanation}. The 
 The parameters $\lambda$ in \eqref{eq:cert1} arise from the construction of $\sigma$ as Gibbs state as in \eqref{eq:gibbs1}.  Moreover the gradient has the  form 
 \begin{align}\label{eq:grad}
     G_\sigma= \log\sigma-\log \Phi(\sigma) 
 \end{align}
 which we will convenient in the construction of rational certificates later on. 
In Sec.~\ref{sec:outer-certificate} we show that for $\sigma=\rho^*$ the certificate \eqref{eq:cert1} nicely recovers the corresponding KKT conditions of \eqref{eq:qkdAB}. It is therefore tight at the optimum and also becomes increasingly tight as the candidate approaches it.
Thus, certification requires only one additional smallest-eigenvalue computation and yields
\begin{align}\label{eq:certificate_inequality}
F_{\mathrm{cert}}
\leq F^*
\leq F_{\mathrm{cand}} .
\end{align}
The resulting gap $\varepsilon_{num}=F_{\mathrm{cand}}-F_{\mathrm{cert}}$ provides a direct and rigorous measure of the remaining numerical uncertainty. Moreover, \eqref{eq:certificate_inequality} can be directly used for a rational approximation of the resulting key rate. That is a certified bound given by a rational number as a lower bound, which can be directly incorporated into a security analysis. Phase 1 and Phase 2 can be repeated until a prescribed target precision is reached. In the detailed convergence analysis, see Sec.~\ref{sec:outer-convergence}, we show that this method converges sublinearly with $O(1/k)$ in the worst case and much faster in practice.  

From a computational perspective, both phases require only standard spectral operations on the matrices of the original Hilbert-space of dimension $d$. In the candidate phase, the dominant operation is the diagonalization of the $d\times d$ effective Hamiltonian, while the Newton update itself and acts on an $n\times n$ Hessian, with $n$ the number of constraints. The certificate adds only a smallest-eigenvalue computation for another $d\times d$ matrix. For a fixed number of constraints, the working memory therefore scales as $O(d^2)$, while a straightforward dense implementation's runtime is dominated by $O(d^3)$ spectral decompositions.

\subsection{Benchmarking}\label{sec:bench}
In Fig.~\eqref{fig:runtime_comparison}, we first compare our method with the benchmarks reported by Lorente et al.\cite{Lorente2025}.  There the authors provide, beside others, benchmarking results of a Gauß-Radau based integral approximation\cite{Arajo2023} and an interior point method on the entropy cone \cite{Lorente2025,He2025}. Their example instances for DMCV-QPSK, MUB, and overlapping-basis protocol are curated on a git hub repository accompanying their work. It is the base for our benchmark. 
\begin{figure}
    \centering
    \includegraphics[width=1\linewidth]{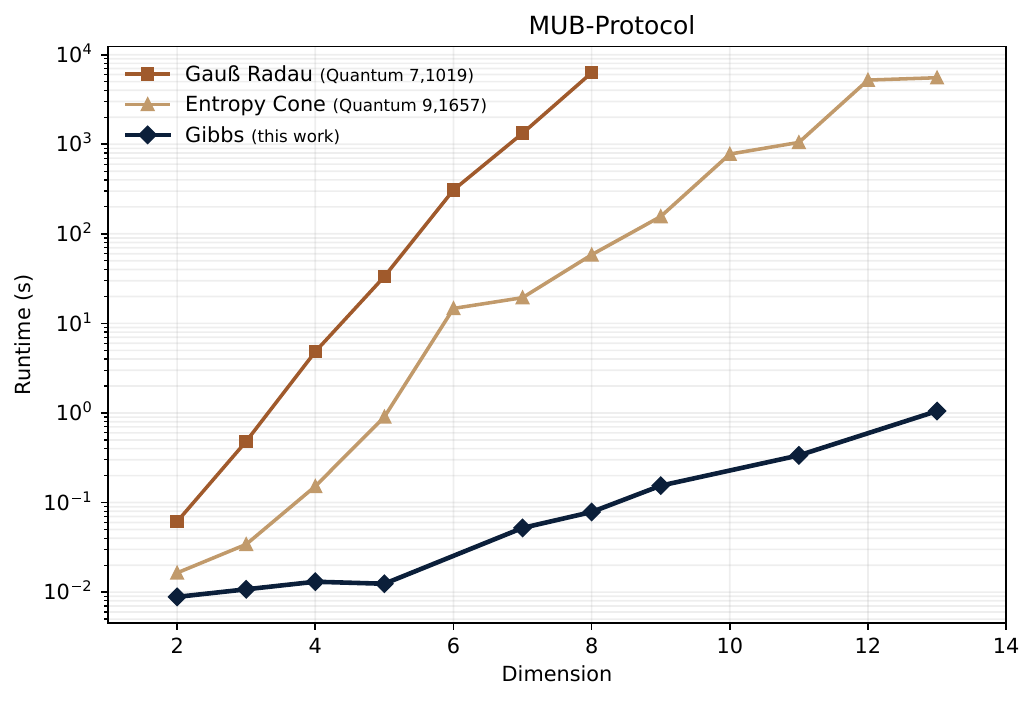}
    \includegraphics[width=1\linewidth]{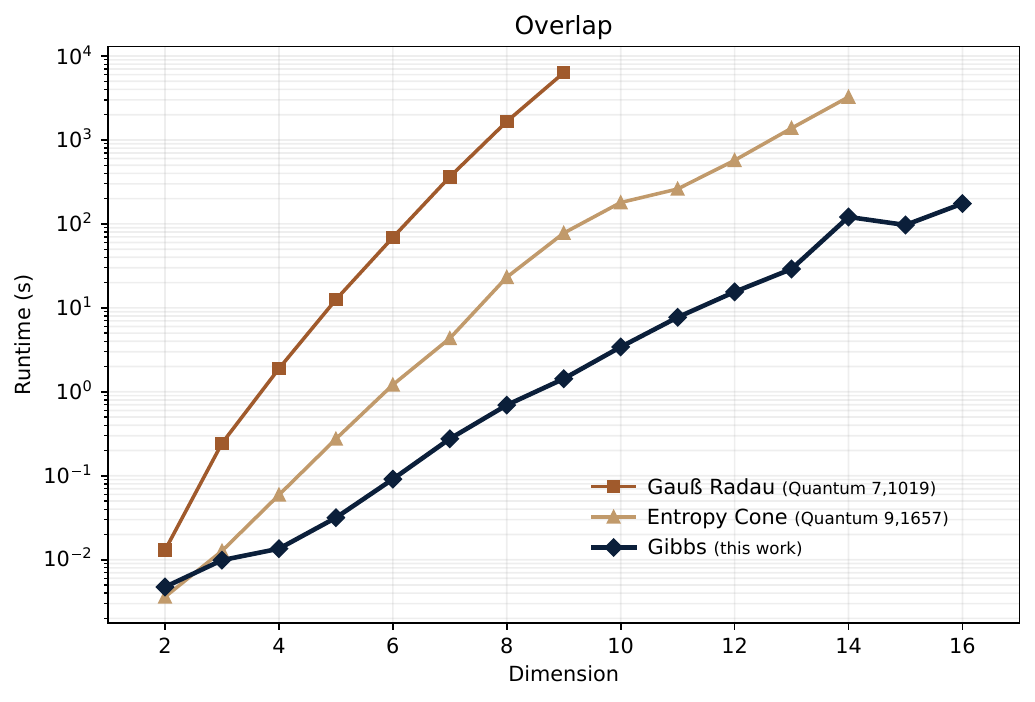}
    \caption{A runtime comparison for the MUB protocol and the overlap protocol from \cite{Arajo2023} as a function of the local dimension $d_A = d_B$ is shown. We compare the computation of the extractable randomness from \eqref{eq:qkdAB} in \cite{Arajo2023}, which uses the Gauß-Radau approximation of the relative entropy \cite{Brown2024}, the entropy cone methods \cite{Lorente2025} and our method. The underlying devices are a laptop and the other hardware described in Fig.~\ref{fig:bild1}(a). All protocols deliver at least precision $10^{-6}$ to the optimizer. Due to the logarithmic scaling, we observe that our new technique provides an improvement of several orders of magnitude.}
    \label{fig:runtime_comparison}
\end{figure}
For DMCV we use $\alpha=0.35$, $\xi=0.05$, and $L=60,\mathrm{km}$, while the MUB and overlap instances use visibility $v=0.95$. Our code is implemented in Python and runs on a standard laptop. As numerical target we require a candidate--certificate gap of at least $\varepsilon_{\rm num}\leq10^{-5}$ nats. The reference timings of Lorente et al.\cite{Lorente2025} were obtained with different software and  workstation hardware and are not normalized to the same precision criterion. Figures \ref{fig:bild1}(a) and \ref{fig:bild2} therefore do not define a hardware-independent speedup factor, but they clearly show that the same benchmark families can be treated with substantially smaller computational resources by our method.

To test whether this remains true on embedded hardware, we run the DMCV-QPSK\cite{Lorente2025} benchmark on a Raspberry Pi 3. The Pi calculations use one BLAS thread and the same $10^{-6}$-nat certificate target. Here we use a Julia implementation with a more memory efficient newton step. The  warm-up time for loading an instance is excluded and the median of five runs is reported. We compare to our Python based reference implementation on a laptop. For completeness reported timings from Lorente et. al\cite{Lorente2025} are also depicted in Figure \ref{fig:bild1}(a).
On the raspberry Pi we see an expected increase in runtime, but no qualitative change of the scaling with $N_c$ over the tested range. Since the laptop and the Pi runs use different software stacks, their ratio is not a pure processor benchmark. The relevant result is that the complete candidate-and-certificate calculation remains feasible on Cortex-A53 based hardware.
\begin{figure}
    \centering
    \includegraphics[width=1\linewidth]{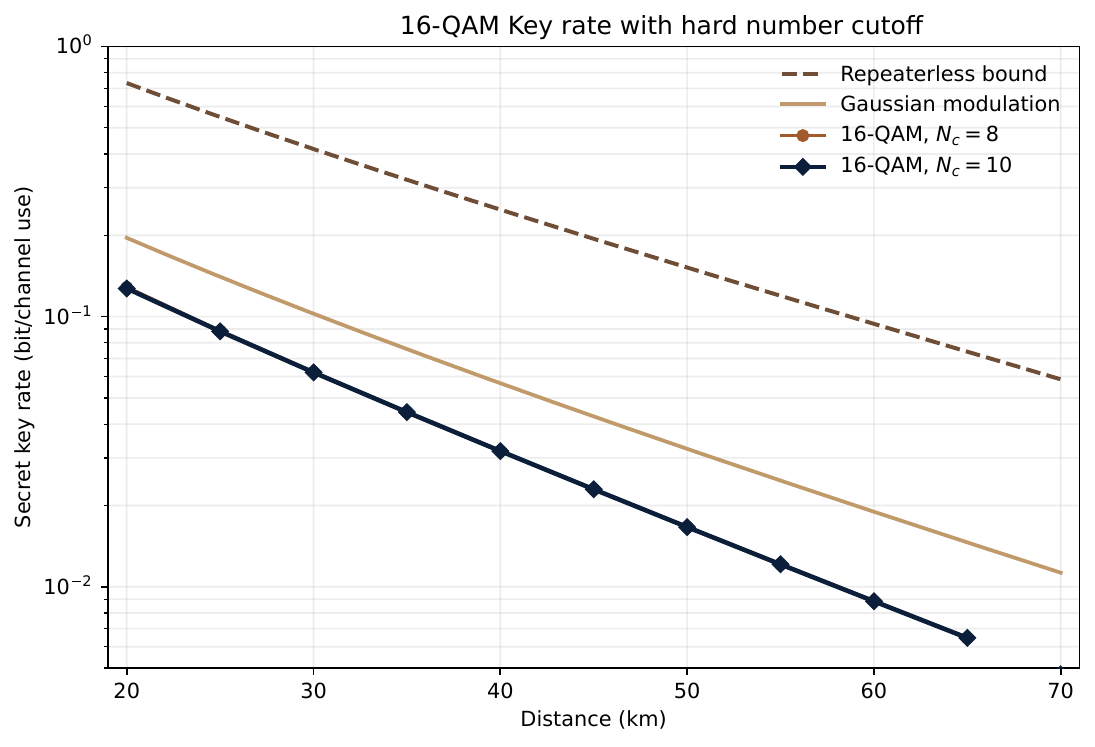}
    \caption{A discrete modulation CVQKD protocol\cite{Denys2021} with 16-QAM coding is shown. We take parameters $V_A=5$, $\nu=0.085$, $\xi=0.02$, $\beta=0.95$, and fiber loss $0.2\,\mathrm{dB/km}$. We use hard photon-number cutoffs $N_c=8$ and $N_c=10$ and evaluate distances up to $90\,\mathrm{km}$. The target certificate gap is $10^{-6}$ nats, with a fixed iteration budget and a runtime limit of four minutes per point. For comparison, we also depict the repeaterless bound\cite{pirandola2017fundamental} and the rate achievable by a Gaußian modulation protocol\cite{Denys2021,navascues2006optimality}. }
    \label{fig:bild3}
\end{figure}
As a more demanding test, presented in Fig.\ref{fig:bild3}, we test a discrete modulation CVQKD protocol\cite{Denys2021} with 16-QAM coding and hard number cutoff. Albeit the effective dimension, which is $d=144$ for $N_c=8$ and $d=176$ for $N_c=10$ stays moderate, we have 319 constraints that set a real numerical challenge to our newton methods. It manifests in our computation as an increase in runtime amount for the newton steps, such as a slower convergence rate for the certificate.  %At the smallest key rates the candidate remains stable, while the certificate gap increases and no longer closes within this budget. 
Nonetheless, even on a laptop instances of this size are computationally accessible. Although a rigorous dimension reduction is still required for a complete protocol analysis, we can already see behaviors that could be promising when properly verified. On one hand the error between a number cutoff 8 and 10 is on are between $10^{-3}$ and $10^{-6}$ and hence to small to be visible in Fig.\ref{fig:bild3}. So we can hope to also get good results after correcting for dimension reduction error on that scale. 
On the other, the scaling of the key rate with the distance is promising. It suggest a substantial advantage of this coding compared to Qpsk and especially existing analytical estimates\cite{Denys2021}. If this advantage will persist in a setting without hard number cutoff 16-QAM could be an attractive coding scheme for practical CV-QKD systems. 

\subsection{Rational approximation}
The computations presented here are carried out using finite-precision floating-point arithmetic, as is customary in numerical optimization. Modern numerical algorithms are designed to keep the resulting rounding errors under control, but they do not, by themselves, provide an exact final bound. For rigorous cryptographic security statements, this distinction can become relevant: even very small numerical errors may matter when key rates are small and protocols operate over extremely large block lengths. This is particularly relevant because key rates can already be very small, while practical protocols may operate over extremely large block lengths. For examples, values around ($10^{12}$) signals are not unrealistic\cite{Liu2023_1002km_TFQKD,Matsuura2021_FiniteSizeCVQKD}. Hence, even tiny numerical errors can accumulate into a non-negligible uncertainty that can dilute the final security claim. Although this issue is well recognized, an explicit treatment of numerical imprecision is not ubiquitous in numerical QKD calculations. A notable exception is the approach of Winick et al.\cite{Winick2018}.

One way to handle this is to perform at least the final certification step in exact rational arithmetic. For many numerical approaches this is prohibitively difficult. For our method, this is surprisingly straightforward.

We briefly outline a path for amending our method with rational approximations. Assume that both the constraints $M$ and the map $\Phi$ are given in rational numbers. 
Note that, $\sigma$ in the certificate \eqref{eq:cert1}  is actually not required to be a feasible state. The bound in fact holds\footnote{ a bad choice of $\sigma$ will however give a loose certificate} for any $\sigma\geq 0$. Let $H\in\mathbb{Q}^{d\times d}$ be a rounding of effective Hamiltonian in the last candidate step and let $\lambda\in\mathbb{Q}^n$ be a rounding of the corresponding parameters. Set $\sigma=e^{H}\geq0$. Any $\beta$ for which
\begin{align}\label{eq:beta}
    \Phi[e^H]\leq \beta e^{H-\lambda M} 
\end{align}
holds gives an immediate bound on $F$. For this, take the $log$ on both sides of \eqref{eq:beta} and plug this into the certificate \eqref{eq:cert1} using \eqref{eq:grad}. We get 
\begin{align}
    F^*\geq \log(1/\beta) - \lambda m 
\end{align}
It hence only remains to find a feasible $\beta \in \mathbb Q$. In  \eqref{eq:beta} , this can be easily done by bounding $\Phi(e^H)$ from above and $e^{H+\lambda M}$ from below by finite exponential expansion and appropriate error terms\cite{Higham2005,AlMohyHigham2011}. The validity of a number $\beta$ can be checkd analytically by Cholesky decomposition. 

In our implementation we represent numbers on a dyadic grid, this is we restrict to rational numbers of the form \(k/2^b\), so that every arithmetic operation can be tracked exactly with an explicit rounding radius.
With sufficiently fine approximations of the exponential\cite{Johansson2017Arb} in \eqref{eq:beta} the error resulting from the rational approximation scales only with the grid size $b$ as $\approx O(2^{-b})$. A quantitative behavior on benchmarking instances from instance\cite{Lorente2025} can be found in Fig.~\ref{fig:bild4}. It confirms with an exponential decrease of the precision overhead due to rational approximation in the bit leth $b$.
\begin{figure}
    \centering
    \includegraphics[width=0.9\linewidth]{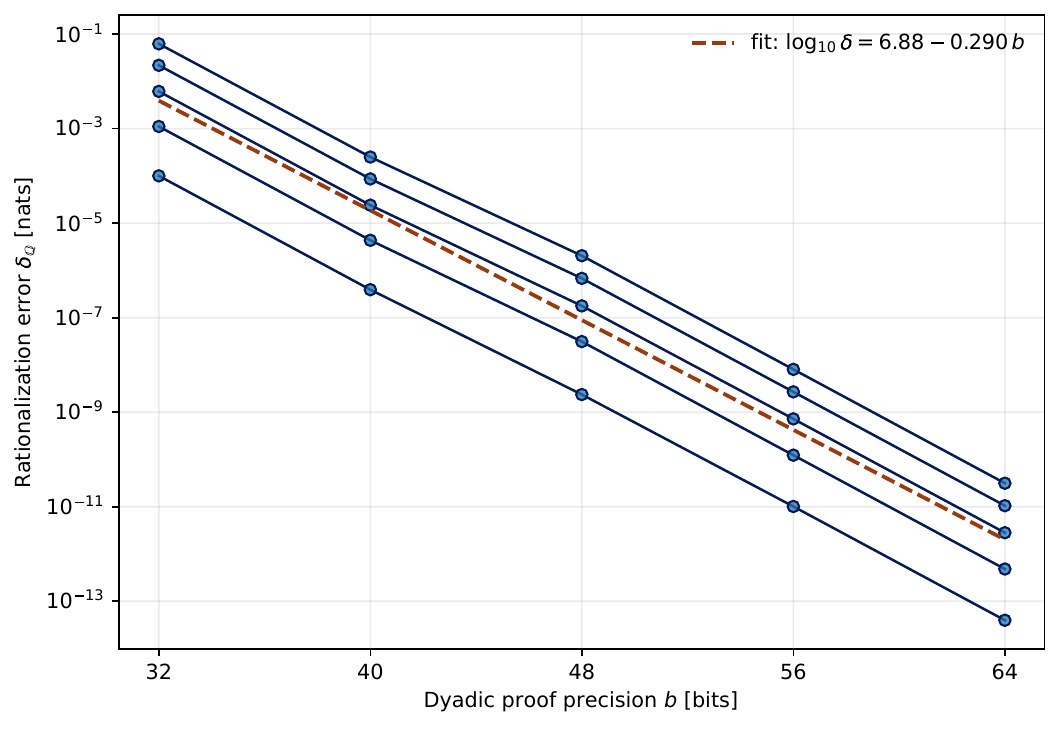}
    \caption{Benchmark of the precision of the rational approximation. We test different instances $(d=2\dots 8)$ of the overlap protocol\cite{Lorente2025} for different dyadic bit sizes $b$ and compare to the result of pure floating point numeric. Already 64 bits, which can be achieved with reasonable computation effort, are sufficient to reduce the extra error to an size of $10^{-9}$ to $10^{-13}$ . }
    \label{fig:bild4}
\end{figure}
\section{Discussion}\label{sec:discussion}
The first methods for computing key rates were based on suboptimal lower bounds to  conditional entropy. A direct path, which is still used today on big instances, is to estimate it by the min-entropy, which can be computed via a short SDP. A better variational bound was found by Coles et al.\cite{Coles2016}, this ansatz can also be extended to general algebraic qkd settings\cite{Tan2021}. Its basic working mechanism originates in proof techniques for entropic uncertainty relations\cite{Abdelkhalek2016,Schwonnek2018} and shares similarities to the Gibbs state optimality step used in this work. 
Modern methods for computing keyrates can be sorted into three branches. The first are methods  based on integral approximations of the relative entropy resulting in SDP hierarchies. This can either be done by Gauß-Radau quadratures \cite{Brown2024,Arajo2023} or 
via the layer cake representation\cite{Frenkel2023,KossmannDIbounds2024,Kossmann2026_commuting_operator}. Whereas the latter\cite{Kossmann2026} is more flexible and has better numerical performance. In the second branch, the key rate computation is directly casted as an optimization problem on the entropy cone\cite{Lorente2025,Fawzi2023,He2025}. Up to now the methods using the entropy cone had the best numerical performance, which is why we used them as comparisons for our numerical benchmarks. 
Our new method can be located in the third branch. The combination of finding a good candidate and certifying it via a gradient-based certificate was also used by Winick et al.\cite{Winick2018} and subsequent works\cite{George2021,Hu2022}, which are now part of the Openqkdsecurity package\cite{burniston2026openqkdsecurity}. In contrast to this work, both the Frank-Wolfe algorithm for candidate finding and the certificate require to solve semidefinite optimization problems an suffer from the same $d^6$ memory bottleneck as other methods. 

To put this general problem into perspective, consider for example the problem of estimating the conditional von Neumann entropy for a bipartite state $\rho_{AB}$ acting on a $d_A d_B$-dimensional Hilbert space subject to the statistics of tomographically complete measurements on both sides. For example arising from  a \emph{mutually unbiased bases protocol} (MUB).
The size of the
resulting optimization problem is governed by the dimension $d_A d_B$ and by the number of constraints $n$. When Alice and Bob each employ $d_A+1$ and $d_B+1$ MUB's, the number of raw constraints is
\begin{align}
  n = d_A(d_A+1)\,d_B(d_B+1) = \Theta\!\left(d_A^2 d_B^2\right),
\end{align}
so that $m$ grows quadratically in the local dimensions. Consequently, a generic interior-point treatment that forms and factorizes the associated Schur complement incurs a memory cost of $\Theta(n^2) = \Theta(d_A^4 d_B^4)$ and a per-iteration runtime cost dominated by $\Theta(n^3) = \Theta(d_A^6 d_B^6)$ (e.g. Cholesky factorization), which quickly becomes prohibitive even for moderate local dimensions independently of the specific technique for handling the non-linear conditional von Neumann entropy as a cost function\cite{Nesterov2018}.

The methods presented in this work are not only better in computation performance, they also share several practical properties that make them particularly attractive for implementations in security-critical environments.

\textbf{(i) Hardware-efficient and deployable.}
The numerical core consists almost entirely of standard spectral operations and requires only quadratic memory in the Hilbert-space dimension, making it well suited to established BLAS and LAPACK implementations. We demonstrate this by running the complete solver on a Raspberry Pi~3 B+ with $1,\mathrm{GB}$ of memory and an Arm Cortex-A53 processor. To put these resources into perspective, processors of this class are found in embedded devices ranging from home routers to typical control boards used in existing qkd devices. A complete numerical security analysis can therefore be performed within a hardware regime representative of widely deployed embedded systems.

\textbf{(ii) Small and auditable implementation.}
The simplicity of the numerical primitives allows the complete algorithm to be implemented in fewer than 100 lines of Common Lisp, with eigensystem decomposition as the only sophisticated numerical primitive at the algorithmic level. While compact code does not by itself guarantee software security, it substantially reduces the amount of code that must be trusted, reviewed, and maintained.

\textbf{(iii) Verifiable final certificate.}
The final certificate can be evaluated using controlled rational approximations and exact arithmetic with little additional overhead. Its validity can therefore be established independently of floating-point accuracy or the numerical behaviour of the preceding optimization.

For future work, the substantial reduction in computational cost opens several directions, especially for protocol design and optimization. In discrete-modulated CV-qkd, for example, higher-order constellations such as 16-QAM or 64-QAM can now be explored systematically and compared with Gaussian modulation\cite{Denys2021}, while pulse and energy shaping can be incorporated as additional optimization parameters. The increased numerical headroom may likewise make more refined dimension-reduction techniques practically accessible. More fundamentally, key-rate estimation becomes fast enough to be included directly in a feedback loop: experimental parameters, modulation regions, or pulse shapes could now be optimized adaptively against the certified key rate itself.

As a last point, we have to stress that the so far best reported finite size \cite{Arqand2025} analyses are based on the sandwiched Renyí entropy\cite{Tomamichel2016}. For future work, our method can be easily adapted to this without changing the basic architecture or sacrificing too much of its performance.

\section{Methods}\label{methods}

\subsection{Formulation as bipartite relative entropy optimization} It is well known that the qkd problem \eqref{eq:qkdproblem} can be formulated as an optimization acting only on the  Alice--Bob system \cite{Tan2021,Tomamichel2016}. This will be our starting point.  After source replacement and a justified finite-dimensional reduction, we recall the feasible set of states from \eqref{eq:feasible-set}. 
%Here $M_i$ are Hermitian observables and the $m_i$ are the observed expectation values. They can represent parameter estimation data, calibration information that Alice and Bob use to build the set of compatible states. 
Furthermore, let $\cG_0: \cB(\cH)\to \cB(\cK_0)$ be the CPTP preprocessing map and let
\begin{equation} \label{eq:direct-pinching}
\Phi_0(Y)\coloneqq\sum_a P_a Y P_a
\end{equation}
when $\supp(\mu) \subseteq \supp(\nu)$, and set it to $+\infty$ otherwise. Within this notation the qkd problem is to compute
\begin{equation} \label{eq:normalized-qkd-problem}
    F^* \coloneqq \min_{\rho \in \cF} D(\cG_0(\rho)\,\|\, \Phi_0(\cG_0(\rho)).
\end{equation}
Taking $\cG_0$ to be the identity recovers the direct pinching problem. For now this setting already covers a lot of qkd cases, it does not hold yet for the general case involving post-selections and thus trace-nonincreasing operations. We leave that extension to future work. 

In Appendix~\ref{app:channel-reduction}, we show how to avoid running into ill-defined logarithms, which leads to CPTP maps that are compressed onto a nonzero support, $\cG$, $\Phi$, and their combination $\cZ \coloneqq \Phi \circ \cG$. They satisfy, on their declared output spaces, $\cG(\rho) \succ 0$, $\cZ(\rho) \succ 0$ for $\rho \succ 0$. Using $S(X) = -\tr(X \log X)$ with the supported logarithm and $0 \log 0 \coloneqq 0$, together with the pinching entropy identity gives
\begin{equation}
  D\!\left(\cG_0(\rho)\middle\Vert\Phi_0(\cG_0(\rho))\right)
  =S(\cZ(\rho))-S(\cG(\rho)).
  \label{eq:compressed-objective-identity}
\end{equation}
Thus
\begin{equation}
  \begin{aligned}
    F(\rho)&\coloneqq S(\cZ(\rho))-S(\cG(\rho)),\\
    F^*&\coloneqq \min_{\rho\in\mathcal F}F(\rho).
  \end{aligned}
  \label{eq:channel-objective}
\end{equation}
Denote the optimizer set by
\begin{equation}
  \mathcal S^*\coloneqq\operatorname*{argmin}_{\rho\in\mathcal F}F(\rho).
  \label{eq:optimizer-set}
\end{equation}
The normalized Discrete-Modulated Continuous Variable (DM CV) model of Gonz\'alez Lorente et
al.~\cite{Lorente2025} has an isometric preprocessing map. Identifying its support with the input space reduces the objective to
\begin{equation}
  F(\rho)=S(\cZ(\rho))-S(\rho).
  \label{eq:dmcv-channel-objective}
\end{equation}
The compact block form of $\cZ$ and its Naimark construction are given in Appendix~\ref{app:channel-reduction}. Direct pinching is the further special case $\cG = \operatorname{id}$ and $\cZ = \Phi$.

Throughout the convergence analysis we assume that $cF$ contains a specified Slater state $\rho_\text{sl} \succ 0$, that $\{\mathbbm{1}, M_1, \dots, M_n\}$ is linearly independent, and that the initial state is feasible $\rho_0 \in \cF$ with $\rho_0 \succ 0$. We use the entropies with supported logarithms and also assume exact affine projections, i.e. that the inner Gibbs solver converges exactly (which is true asymptotically).

\subsubsection{Basic calculus for one step of improving a candidate}\label{sec:outerloop}
The main part of our method is to approximate \eqref{eq:channel-objective} by an iteration of inner points $\rho_k\in\mathcal{F}$ with increasing solution quality. At outer iteration $k$, fix the current state $\sigma\coloneqq\rho_k \in \cF$ with $\sigma \succ 0$, and consider
\begin{equation}
  \begin{aligned}
    T(\sigma)
    &\coloneqq\operatorname*{argmin}_{\rho\in\mathcal F}
      D(\rho\Vert\chi_\sigma),\\
    \rho_{k+1}&\coloneqq T(\rho_k).
  \end{aligned}
  \label{eq:exact-outer-map}
\end{equation}
Here $\sigma$, and hence the reference state $\chi_\sigma$, is held fixed while $\rho$ is varied over $\cF$. We call \eqref{eq:exact-outer-map} the exact outer update. The following subsection shows how its state space minimization is evaluated through an inner optimization over Gibbs multipliers.

It remains to construct $\chi_\sigma$ and verify this update improves the objective. First define the objective gradient
\begin{equation}
  G_\sigma\coloneqq\nabla F(\sigma)
  =\cG^*(\log\cG(\sigma))-\cZ^*(\log\cZ(\sigma)),
  \label{eq:channel-gradient}
\end{equation}
which also obeys $\tr(\sigma G_\sigma) = F(\sigma)$.
The moving reference is 
\begin{equation}
  \begin{aligned}
    H_\sigma&\coloneqq\log\sigma-G_\sigma,\\
    Z_\sigma&\coloneqq\tr e^{H_\sigma},\\
    \chi_\sigma&\coloneqq\frac{e^{H_\sigma}}{Z_\sigma}.
  \end{aligned}
  \label{eq:moving-reference}
\end{equation}

To justify the update, the relative entropy perturbation of $F$ gives the majorization
\begin{equation}
  F(\rho)\le
  F(\sigma)+\tr[G_\sigma(\rho-\sigma)]+D(\rho\Vert\sigma) =: Q_\sigma(\rho),
  \label{eq:tangent-majorizer}
\end{equation}
which we derive in Appendix~\ref{app:channel-bregman} by the Bregman identity combined with data processing under $\Phi$ and $\cG$. This line of reasoning also establishes convexity of $F$.
Completing this expression as a relative entropy gives
\begin{equation}
  D(\rho\Vert\chi_\sigma)
  =Q_\sigma(\rho)+\log Z_\sigma.
  \label{eq:projection-majorizer-equivalence}
\end{equation}
As the second term is independent of $\rho$, $T(\sigma)$ also minimizes $Q_\sigma$ over $\cF$. Given that $\sigma$ itself is feasible, we have $F(T(\sigma)) \leq Q_\sigma(T(\sigma)) \leq Q_\sigma(\sigma) = F(\sigma)$. Consequently, every exact step improves the objective. For direct pinching, $G_\sigma = \log \sigma - \log\Phi(\sigma)$ and the update simplifies to
\begin{equation}
  \begin{aligned}
    \chi_\sigma &=\Phi(\sigma),\\
    \rho_{k+1}&=\operatorname*{argmin}_{\rho\in\mathcal F}
      D(\rho\Vert\Phi(\rho_k)).
  \end{aligned}
  \label{eq:direct-outer-map}
\end{equation}
The pinching Pythagorean identity in Lemma~\ref{lem:pinching-identities} then gives the especially transparent comparison
\begin{equation}
  F(\sigma)\ge D(T(\sigma)\Vert\Phi(\sigma))
  \ge F(T(\sigma)).
  \label{eq:direct-candidate-improvement}
\end{equation}

\subsection{Minimization of the relative entropy in each step}\label{sec:affine-gibbs-projection}
The central simplification for tackling optimization step \eqref{eq:exact-outer-map} stems from the
structure we derive here. For $\lambda = (\lambda_1, \dots, \lambda_n) \in \mathbb{R}^n$, write $\lambda \cdot M \coloneqq \sum_i \lambda_i M_i$ and $\lambda \cdot m = \sum_i \lambda_im_i$. For a positive reference state $\chi$, define
\begin{align}
  Z_\chi(\lambda)&\coloneqq\tr \exp(\log\chi-\lambda\cdot M),
  \notag\\
  \rho_{\chi,\lambda}&\coloneqq\frac{\exp(\log\chi-\lambda\cdot M)}{Z_\chi(\lambda)},
  \notag\\
  g_\chi(\lambda)&\coloneqq\log Z_\chi(\lambda)+\lambda\cdot m.
  \label{eq:gibbs-dual-objective}
\end{align}

\paragraph{Gibbs form of the exact step.} The unique minimizer of $D(\rho\,\|\,\chi)$ over $\cF$ is
\begin{equation}
  \rho^*_\chi=\rho_{\chi,\lambda_\chi^*},
  \qquad
  \lambda_\chi^*\coloneqq \operatorname*{argmin}_{\lambda\in\mathbb R^n}
  g_\chi(\lambda).
  \label{eq:unique-information-projection}
\end{equation}
It is positive definite and matches every affine moment set by the constraints exactly. Thus the state space optimization in one outer step becomes an unconstrained problem in only $n$ real Lagrange multipliers. The operator $-\log \chi + \lambda \cdot M$ may be read as an effective Hamiltonian, with $Z_\chi$ its partition function and $\rho_{\chi, \lambda}$ its Gibbs state. The foral statement and proof are in Lemma~\ref{lem:gibbs-projection} of Appendix~\ref{app:gibbs-projection-proof}. 

For every $\tau \in \cF$, including singular $\tau$, the exact Gibbs projection satisfies the affine Pythagorean identity of Proposition~\ref{prop:affine-pythagorean}; see also Ji~\cite{ji2022matrix}.
\begin{equation}
  D(\tau\Vert\chi)
  =D(\tau\Vert\rho^*_\chi)
   +D(\rho^*_\chi\Vert\chi).
  \label{eq:affine-pythagorean-main}
\end{equation}
The multiplier objective is real analytic, and its gradient has the direct interpretation \begin{equation}
  \bigl(\nabla g_\chi(\lambda)\bigr)_i
  =m_i-\tr(\rho_{\chi,\lambda}M_i),
  \label{eq:dual-gradient}
\end{equation}
i.e. it is precisely the vector of moment mismatches. Its Hessian is the positive response matrix of these moments, in other words, the Gram matrix of the observables in the inner product in the underlying Bogoliubov-Kubo-Mori (BKM) geometry. Together with the fact that $\{\mathbbm1,M_1,\ldots,M_n\}$ is linearly independent therefore makes the Hessian positive definite at every finite multiplier. Consequently, the Newton direction is well defined by 
\begin{equation}
  \nabla^2g_\chi(\lambda^{(t)})\,\Delta\lambda^{(t)}
  =-\nabla g_\chi(\lambda^{(t)}).
  \label{eq:newton-direction}
\end{equation}
Together with the coercivity proved in Appendix~\ref{app:gibbs-projection-proof}, this strict convexity guarantees that $g_\chi$ has a unique global minimizer, the Gibbs state, characterized by $\nabla g_\chi = 0$ and hence by satisfying all constraints. Its analytic gradient and Hessian make standard smooth convex solvers, like the damped Newton method we use, ideal. The eigensystem used to form the Gibbs state is also used to form the Hessian, so a Newton step needs one matrix diagonalization and one $n \times n$ linear solve. In previous works, such an optimization is usually done by semi definite programming and quickly becomes resource intensive. The essential ingredient for our efficiency improvement is the found Gibbs structure of the inner optimizers. For more details about the underlying BKM geometry and how we use it for these results, see Appendix~\ref{app:bkm-geometry}.

\subsection{Outer certificates}\label{sec:outer-certificate}
The candidate iteration alone approaches the optimum from above: every exactly feasible candidate $\widehat{\rho}$ satisfies $F^* \leq F(\widehat{\rho})$. However, for a qkd security calculation this is not enough, because a secure key rate requires a certified lower bound on the entropy contribution $F^*$. We therefore accompany each candidate by an outer certificate below $F^*$. Together, the two values bracket the optimum and quantify the remaining optimization gap. 

To construct the certificate, we can again use geometric reasoning and our gradients. Let $\sigma \succ 0$ be a linearization state, typically the current outer iterate, and set $G_\sigma = \nabla F(\sigma)$. Geometrically, convexity means that the tangent to $F$ at $\sigma$ lies below $F$, 
\begin{equation}
  F(\rho)
  \ge F(\sigma)+\tr\!\left[G_\sigma(\rho-\sigma)\right]
  =\tr(G_\sigma\rho),
  \label{eq:supporting-linearization}
\end{equation}
where the last equality uses $\tr(\sigma G_\sigma)=F(\sigma)$. Minimizing this tangent over $\cF$ would already give a valid lower bound, but it is itself a linear semidefinite program. Solving that problem at every iteration would forfeit the main computational advantage of the Gibbs method. Instead, shift $G_\sigma$ by an arbitrary affine combination of the moment observables. Feasibility fixes the expectation of this shift, while the expectation of the resulting Hermitian matrix is bounded below by its smallest eigenvalue. Hence every $\alpha \in \mathbb{R}^n$ gives the inexpensive certificate
\begin{equation}
  F_{\mathrm{cert}}(\sigma,\alpha)
  \coloneqq \lambda_{\min}\!\left(G_\sigma+\alpha\cdot M\right)
    -\alpha\cdot m
  \le F^*.
  \label{eq:certificate-pencil-bound}
\end{equation}
The multiplier already produced by the Gibbs step is a natural choice for $\alpha$, although validity does not depend on that choice. Any exactly feasible candidate $\widehat{\rho}$ closes in on the other side
\begin{equation}
  F_{\mathrm{cert}}(\sigma,\alpha)
  \le F^*\le F(\widehat\rho).
  \label{eq:candidate-certificate-bracket}
\end{equation}
At a positive-definite optimizer the affine optimality condition makes the certificate pencil a scalar multiple of the identity, so an appropriate multiplier closes the bound exactly, making our certificate tight. This is the affine equality specialization of Ref.~\cite{Winick2018}. Proposition~\ref{prop:candidate-certificate} in Appendix~\ref{app:certificate-proof} gives the formal statement and proof.

\subsection{Outer convergence} \label{sec:outer-convergence}
Let us consider the exact sequence in \eqref{eq:exact-outer-map}. Every finite iterate is positive and feasible, and the objective values decrease monotonically. The key estimate compares one step with any optimizer $\rho^* \in \cS^*$:
\begin{equation}
  F(\rho_{k+1})-F^*
  \le D(\rho^*\Vert\rho_k)-D(\rho^*\Vert\rho_{k+1}).
  \label{eq:one-step-telescoping-bound}
\end{equation}
Summing over all $k$ immediately gives the main convergence result for the objective function
\begin{equation}
  0\le F(\rho_n)-F^*
  \le\frac{D(\rho^*\Vert\rho_0)}{n}
  \le\frac{-\log\lambda_{\min}(\rho_0)}{n}.
  \label{eq:global-objective-rate}
\end{equation}
Accordingly, the candidate value converges to the optimum at worst $1/ n$. Compactness and continuity also imply that the states approach the optimizer set 
\begin{equation} \label{eq:state-convergence-qualitative}
    \inf_{\tau \in \mathcal{S}^*} \|\rho_n - \tau\|_1 \longrightarrow 0.
\end{equation}
If the optimizer is unique, the whole sequence converges to it even when it is singular. Theorem~\ref{thm:global-convergence} in Appendix~\ref{app:global-convergence-proof} gives the formal statement and the complete relative entropy proof. For isometric preprocessing its numerator sharpens to $D(\cZ(\rho^*)\Vert\cZ(\rho_0))$, and direct pinching is obtained by setting $\cZ=\Phi$.

In general, however, we cannot assume that the optimizer is unique. In this case, the above result only proves that the state converges to a set, but not necessarily to either of the optimizers in the set. If $\cS^*$ contains at least one positive-definite state $\rho^\sharp \succ 0$, the relative entropy estimate in \eqref{eq:one-step-telescoping-bound} keeps the exact trajectory uniformly away from the boundary. The complete sequence then converges to one of the positive-definite optimizers $\rho^\flat \succ 0$ selected by $\rho_0$, and the exact multipliers converge as well. These selection results and their proofs are given in Corollary~\ref{cor:interior-selection} of Appendix~\ref{app:interior-selection-proof}.

\subsubsection{Geometric convergence}
\label{sec:geometric-convergence}
Once an optimizer $\rho^* \in \cS^*$ has been selected by the trajectory, we can study at which rate do we reach that state with our algorithm. Near the optimizer, the second-order perturbation of relative entropy once again reveals the local BKM geometry. Linearizing the exact map $T$ \eqref{eq:exact-outer-map} in that geometry separates directions along the optimizer set from transverse directions. Our results in Appendix~\ref{app:bkm-geometry} show that the former directions are flat, whereas the latter ones are \textit{strictly} contracting. Unless every feasible state is already optimal, the transverse linearized factor is strictly below one. Consequently, for a suitable $0 < \theta < 1$ and constants $c_\theta, \tilde{c}_\theta$, we see a geometric contraction
\begin{align}
  \norm{\rho_k-\rho^*}
  &\le c_\theta\theta^k,
  \label{eq:geometric-state-rate}\\
  0\le F(\rho_k)-F^*
  &\le\widetilde c_\theta\theta^{2k},
  \label{eq:geometric-objective-rate}
\end{align}
where the norm holds in any fixed norm on the finite-dimensional state space, like trace norm or the BKM norm that was used. For more details, please see Theorem~\ref{thm:geometric-convergence}. Although this is only a local contraction result, close to the optimizer, Theorem~\ref{thm:global-convergence} and Corollary~\ref{cor:interior-selection} ensure that, whenever $\mathcal{S}^*$ contains a positive-definite state, the complete exact trajectory converges to a selected optimizer and therefore enters the contraction neighborhood after finitely many outer steps. Thus the global $\cO(1 / n)$ objective bound applied from the outset, while the geometric state rate and the $O(\theta^{2k})$ objective rate govern the eventual tail.

\subsection{Inner Gibbs solver}\label{sec:inner-gibbs-solver}
The preceding results concern the exact outer map. For a fixed reference $\chi = \chi_{\rho_k}$, \eqref{eq:dual-gradient} says that the inner gradient is exactly the moment mismatch. Newton's method is therefore a natural correction: its positive response matrix converts the current mismatches into the multiplier update in \eqref{eq:newton-direction}. We globalize the step by standard Armijo backtracking \cite{Boyd2004}.

When the number of constraints is $r \geq 1$, Theorem~\ref{thm:inner-newton-convergence} proves that for every fixed positive $\chi$ and finite initial multiplier, the exact damped sequence converges to $\lambda_\chi^*$. Once it is sufficiently close, full Newton steps are accepted and convergence is quadratic. The results can be summarized as the following two bounds
\begin{align}
  g_\chi(\lambda)-g_\chi(\lambda_\chi^*)
  &=D(\rho^*_\chi\Vert\rho_{\chi,\lambda}),
  \label{eq:dual-gap-relative-entropy}\\
  \norm{\rho_{\chi,\lambda}-\rho^*_\chi}_1
  &\le\sqrt{2\bigl[g_\chi(\lambda)-g_\chi(\lambda_\chi^*)\bigr]}.
  \label{eq:dual-gap-state-bound}
\end{align}
Thus a dual objective gap of at most $\varepsilon$ guarantees that the approximate Gibbs state returned by the inner solver is within $\sqrt{2\varepsilon}$ in trace norm of the exact Gibbs state.

In practice the inner solve stops after finitely many steps. Unless its final moment vanishes exactly, the returned Gibbs state is not exactly feasible and hence does not equal the outer map analyzed above. A practical stopping rule must therefore be paired with a feasibility correction and an analysis of the resulting inexact outer iteration. 

\section*{Acknowledgements}
GK acknowledges support from the Excellence Cluster - Matter and Light for Quantum Computing (ML4Q-2) and by the European Research Council (ERC Grant Agreement No. 948139). R.S.\ is supported  by the DFG under Germany's Excellence Strategy - EXC-2123 QuantumFrontiers - 390837967 and  SFB 1227 (DQ-mat), the Quantum Valley Lower Saxony, and the BMBF projects CBQD, SEQUIN, Quics and ATIQ.  

\section*{Use of LLM}
The core idea of the algorithm was developed in \cite{Schwonnek2018Diss}. In this extended version, \cite{OpenAI2026ChatGPT} and \cite{claude2026} were used for typographic improvements and for refining the details of the mathematical proofs. The authors take full responsibility for the content. 

\bibliographystyle{naturemag}
\bibliography{lit}

% =====================
% Appendix
% =====================
\clearpage
\onecolumngrid
\appendix
\section*{Appendices}
\section{Channel reduction and support}
\label{app:channel-reduction}

Matrix logarithms in the main algorithm must be evaluated on fixed faithful
output spaces, whereas the unreduced protocol maps may contain null
directions.  This appendix removes those directions by exact support
compression, proves that the compressed maps remain CPTP, and verifies that
the compact DM CV block representation preserves the objective.  Set
\begin{equation}
  \begin{aligned}
    Q&\coloneqq \operatorname{supp}(\cG_0(\mathbbm1_{\mathcal H})),\\
    R&\coloneqq \operatorname{supp}(\Phi_0(\cG_0(\mathbbm1_{\mathcal H})))
  \end{aligned}
  \label{eq:channel-support-projectors}
\end{equation}
as orthogonal projections on $\mathcal K_0$, with $\cG_0$ and $\Phi_0$ the uncompressed maps. Let
$\mathcal K_G\coloneqq \operatorname{ran}Q$ and
$\mathcal K_Z\coloneqq \operatorname{ran}R$, and let
$V_Q:\mathcal K_G\hookrightarrow\mathcal K_0$ and
$V_R:\mathcal K_Z\hookrightarrow\mathcal K_0$ be the canonical isometric
inclusions.  Thus
$V_Q^\dagger V_Q=\mathbbm1_{\mathcal K_G}$,
$V_QV_Q^\dagger=Q$,
$V_R^\dagger V_R=\mathbbm1_{\mathcal K_Z}$, and
$V_RV_R^\dagger=R$.  Define
\begin{align}
  \cG:\mathcal B(\mathcal H)&\longrightarrow\mathcal B(\mathcal K_G),
  &\cG(X)&\coloneqq V_Q^\dagger\cG_0(X)V_Q,
  \label{eq:compressed-preprocessing-map}\\
  \Phi:\mathcal B(\mathcal K_G)&\longrightarrow\mathcal B(\mathcal K_Z),
  &\Phi(Y)&\coloneqq V_R^\dagger
    \Phi_0(V_QYV_Q^\dagger)V_R,
  \label{eq:compressed-pinching-map}
\end{align}
Then, the combined map is given by 
\begin{equation}
  \cZ\coloneqq \Phi\circ\cG:
    \mathcal B(\mathcal H)\longrightarrow\mathcal B(\mathcal K_Z).
  \label{eq:compressed-composed-map}
\end{equation}
For $X\succeq0$,
\begin{equation}
  0\preceq\cG_0(X)
  \preceq\norm{X}_{\mathrm{op}}\cG_0(\mathbbm1_{\mathcal H}).
  \label{eq:preprocessing-support-domination}
\end{equation}
Hence every positive output of $\cG_0$ is supported on $Q$.  Decomposing a
general operator into a linear combination of positive operators gives the
same support statement for the complete range.  Compression by $V_Q$ loses
neither trace nor nonzero spectrum and preserves complete positivity, so
$\cG$ is CPTP.  By definition of $Q$,
$\cG(\mathbbm1_{\mathcal H})\succ0$ on $\mathcal K_G$.

On $\mathcal K_G$ there are constants $0<c\le C<\infty$ such that
\begin{equation}
  cV_QV_Q^\dagger
  \preceq\cG_0(\mathbbm1_{\mathcal H})
  \preceq CV_QV_Q^\dagger.
  \label{eq:preprocessing-support-equivalence}
\end{equation}
Applying the positive map $\Phi_0$ shows that
\begin{equation}
  \operatorname{supp}(\Phi_0(V_QV_Q^\dagger))=R.
  \label{eq:pinching-output-support}
\end{equation}
Every output of $Y\mapsto\Phi_0(V_QYV_Q^\dagger)$ is consequently supported
on $R$.  The second compression again preserves trace and complete
positivity; hence $\Phi$ and $\cZ$ are CPTP and
$\Phi(\mathbbm1_{\mathcal K_G})\succ0$ on $\mathcal K_Z$ and
$\cZ(\mathbbm1_{\mathcal H})\succ0$ on $\mathcal K_Z$.  In particular, if
$\rho\succeq b\mathbbm1_{\mathcal H}$ for $b>0$, then
\begin{equation}
  \begin{aligned}
    \cG(\rho)&\succeq b\cG(\mathbbm1_{\mathcal H})\succ0,\\
    \cZ(\rho)&\succeq b\cZ(\mathbbm1_{\mathcal H})\succ0.
  \end{aligned}
  \label{eq:faithful-channel-outputs}
\end{equation}
Thus all output logarithms at positive iterates are taken on fixed exact
supports.

For the normalized DM CV construction, specialize
$\mathcal H=\mathcal H_A\otimes\mathcal H_B$ and let
$\{E_a\in\mathcal B(\mathcal H_B):a\in\mathsf A\}$ be a POVM, so
$E_a\succeq0$ and $\sum_aE_a=\mathbbm1_{\mathcal H_B}$.  With key register
space $\mathcal H_Z\coloneqq \mathbb C^{|\mathsf A|}$ and basis
$\{\ket a:a\in\mathsf A\}$, set $\mathcal K_0\coloneqq \mathcal H_A\otimes\mathcal H_B\otimes\mathcal H_Z,$
\begin{equation}
  \begin{aligned}
    &V:\mathcal H\longrightarrow\mathcal K_0,
    &V\coloneqq \sum_{a\in\mathsf A}
       \mathbbm1_{\mathcal H_A}\otimes\sqrt{E_a}\otimes\ket a
  \end{aligned}
  \label{eq:appendix-dmcv-naimark-data}
\end{equation}
with projectors $P_a\coloneqq \mathbbm1_{\mathcal H_A\otimes\mathcal H_B} \otimes\ket a\!\bra a$. Then the preprocessing map can be written as $\cG_0(X)\coloneqq VXV^\dagger$, where $V^\dagger V=\sum_a\mathbbm1_{\mathcal H_A}\otimes E_a
=\mathbbm1_{\mathcal H}$.  For
$\mathcal K_a\coloneqq \operatorname{ran}(P_aV)\subseteq\mathcal K_0$, let
$V_a:\mathcal K_a\hookrightarrow\mathcal K_0$ be the inclusion and define
\begin{equation}
  \begin{aligned}
    A_a&\coloneqq V_a^\dagger P_aV:\mathcal H\longrightarrow\mathcal K_a,\\
    \mathcal K_Z&\coloneqq \bigoplus_{a\in\mathsf A}\mathcal K_a.
  \end{aligned}
  \label{eq:appendix-dmcv-block-spaces}
\end{equation}
Since $V_aV_a^\dagger$ projects onto $\operatorname{ran}(P_aV)$,
\begin{equation}
  A_a^\dagger A_a
  =V^\dagger P_aV_aV_a^\dagger P_aV
  =V^\dagger P_aV.
  \label{eq:dmcv-block-completeness-step}
\end{equation}
Summing over $a$ gives
$\sum_aA_a^\dagger A_a=\mathbbm1_{\mathcal H}$, so the compact block map
\begin{equation}
  \cZ:\mathcal B(\mathcal H)\longrightarrow\mathcal B(\mathcal K_Z),
  \qquad
  \cZ(X)\coloneqq \bigoplus_{a\in\mathsf A}A_aXA_a^\dagger,
  \label{eq:appendix-dmcv-compact-channel}
\end{equation}
is CPTP.  Each $A_a$ is onto its declared range, and $\cZ(\rho)$ has the same
nonzero block spectra as $\Phi_0(V\rho V^\dagger)$.  Since $V$ is an
isometry, $S(V\rho V^\dagger)=S(\rho)$, and hence the normalized DM CV
objective is $F(\rho)=S(\cZ(\rho))-S(\rho)$.

\section{Relative entropy identities}
\label{app:relative-entropy-identities}

The convergence analysis repeatedly needs exact entropy decompositions that
remain valid for boundary states.  This appendix establishes the pinching,
channel Bregman, affine projection, and moving reference identities used to
turn individual Gibbs steps into telescoping estimates.

\subsection{Pinching entropy and Pythagorean identities}
\label{app:pinching-identities}

\begin{lem}[Pinching entropy and Pythagorean identities]
  \label{lem:pinching-identities}
  Let $\cH$ be a finite-dimensional Hilbert space and let
  $\{\widehat P_b:b\in\mathsf B\}\subset\mathcal B(\mathcal H)$ be a finite
  family of mutually orthogonal projections summing to
  $\mathbbm1_{\mathcal H}$.  Define the pinching
  $\Phi:\mathcal B(\mathcal H)\to\mathcal B(\mathcal H)$ by
  $\Phi(X)\coloneqq \sum_{b\in\mathsf B} P_b X P_b$.  For every
  $X\in\mathcal B(\mathcal H)$ with $X\succeq0$,
  \begin{equation}
    D(X\Vert\Phi(X))=S(\Phi(X))-S(X),
    \label{eq:boundary-pinching-entropy}
  \end{equation}
  including when $X$ is singular.  If $\tau$ and $\sigma$ are density
  operators on $\mathcal H$ and $\sigma\succ0$, then
  \begin{equation}
    D(\tau\Vert\Phi(\sigma))
    =D(\tau\Vert\Phi(\tau))
     +D(\Phi(\tau)\Vert\Phi(\sigma)).
    \label{eq:pinching-pythagorean-appendix}
  \end{equation}
\end{lem}

\begin{proof}
Let $s$ be the number of nonzero pinching blocks.  The pinching inequality
$X\preceq s\Phi(X)$ implies
$\operatorname{supp}(X)\preceq\operatorname{supp}(\Phi(X))$, so the
relative entropy in \eqref{eq:boundary-pinching-entropy} is finite.  The
supported logarithm $\log_{\operatorname{supp}}\Phi(X)$ is block diagonal.
Self-adjointness of $\Phi$ in the Hilbert--Schmidt pairing therefore gives
\begin{equation}
  \tr\!\left[X\log_{\operatorname{supp}}\Phi(X)\right]
  =\tr\!\left[\Phi(X)\log_{\operatorname{supp}}\Phi(X)\right],
  \label{eq:pinching-log-pairing}
\end{equation}
which proves \eqref{eq:boundary-pinching-entropy} directly.

For the second identity, both $\log\Phi(\sigma)$ and
$\log_{\operatorname{supp}}\Phi(\tau)$ are block diagonal.  Applying the same
self-adjointness relation to each logarithm and expanding the three relative
entropies proves \eqref{eq:pinching-pythagorean-appendix}.  The assumption
$\sigma\succ0$ ensures all second-argument logarithms are ordinary and all
three quantities are finite.
\end{proof}

Taking $X=\cG_0(\rho)$ for a density operator $\rho$ on $\mathcal H$ in
\eqref{eq:boundary-pinching-entropy}, and using invariance of the nonzero
spectra under the exact compressions in Appendix~\ref{app:channel-reduction},
gives
\begin{equation}
  D(\cG_0(\rho)\Vert\Phi_0(\cG_0(\rho)))
  =S(\cZ(\rho))-S(\cG(\rho))=F(\rho).
  \label{eq:appendix-compressed-objective-identity}
\end{equation}
In the direct case, $\mathcal K_0=\mathcal H$, $\cG_0=\operatorname{id}$,
and \eqref{eq:pinching-pythagorean-appendix} applies with
$\Phi=\Phi_0$.

\subsection{Channel Bregman identity}
\label{app:channel-bregman}

\begin{prop}[Channel Bregman identity and relative smoothness]
  \label{prop:channel-bregman}
  Let
  $\cG:\mathcal B(\mathcal H)\to\mathcal B(\mathcal K_G)$ and
  $\Phi:\mathcal B(\mathcal K_G)\to\mathcal B(\mathcal K_Z)$ be the
  compressed CPTP maps in Appendix~\ref{app:channel-reduction}, and let
  $\cZ\coloneqq \Phi\circ\cG$.  For density operators $\rho$ and $\sigma$ on
  $\mathcal H$ with $\sigma\succ0$, the Bregman difference of $F$ is
  \begin{equation}
    \mathfrak B_F(\rho,\sigma)
    =D(\cG(\rho)\Vert\cG(\sigma))
     -D(\cZ(\rho)\Vert\cZ(\sigma))
    \label{eq:appendix-channel-bregman}
  \end{equation}
  and satisfies
  \begin{equation}
    0\le \mathfrak B_F(\rho,\sigma)\le D(\rho\Vert\sigma).
    \label{eq:appendix-relative-smoothness}
  \end{equation}
  Consequently, $F$ is convex on the complete state space.  On positive
  states its Hilbert--Schmidt gradient is
  \begin{equation}
    \nabla F(\sigma)
    =\cG^*(\log\cG(\sigma))-\cZ^*(\log\cZ(\sigma))
    =G_\sigma.
    \label{eq:appendix-channel-gradient}
  \end{equation}
\end{prop}

\begin{proof}
For $X\succeq0$, set
$h(X)\coloneqq \tr[X\log_{\operatorname{supp}}X]=-S(X)$.  At $\sigma\succ0$,
\begin{equation}
  \begin{aligned}
    \mathrm Dh_\sigma[X]&=\tr[X(\log\sigma+\mathbbm1)],\\
    \mathfrak B_h(\rho,\sigma)&=D(\rho\Vert\sigma).
  \end{aligned}
  \label{eq:negative-entropy-bregman}
\end{equation}
Linearity and trace preservation give
\begin{align}
  \mathfrak B_{h\circ\cG}(\rho,\sigma)
  &=D(\cG(\rho)\Vert\cG(\sigma)),
  \label{eq:preprocessing-bregman}\\
  \mathfrak B_{h\circ\cZ}(\rho,\sigma)
  &=D(\cZ(\rho)\Vert\cZ(\sigma)).
  \label{eq:output-bregman}
\end{align}
Since $F=h\circ\cG-h\circ\cZ$, subtraction proves
\eqref{eq:appendix-channel-bregman}.  Data processing under $\Phi$ makes its
right-hand side nonnegative.  Data processing under $\cG$ and nonnegativity
of relative entropy give
\begin{equation}
  \mathfrak B_F(\rho,\sigma)
  \le D(\cG(\rho)\Vert\cG(\sigma))
  \le D(\rho\Vert\sigma).
  \label{eq:channel-relative-smoothness-proof}
\end{equation}
For these standard properties of quantum relative entropy, see
Watrous~\cite{Watrous2018}, Theorems~5.35 and 5.38.  Thus $F$ is convex
on positive states.  Mixing boundary arguments with a fixed positive state
and using finite-dimensional entropy continuity extends convexity and
continuity to the complete state space.

Differentiation gives
\begin{equation}
  \nabla F(\sigma)
  =\cG^*(\log\cG(\sigma)+\mathbbm1_{\mathcal K_G})
   -\cZ^*(\log\cZ(\sigma)+\mathbbm1_{\mathcal K_Z}).
  \label{eq:channel-gradient-before-cancellation}
\end{equation}
The adjoints are unital because the channels are trace preserving, so the
identity terms cancel.  Adjointness then gives
\begin{equation}
  \begin{aligned}
    \tr(\sigma G_\sigma)
    &=\tr[\cG(\sigma)\log\cG(\sigma)]\\
    &\quad-\tr[\cZ(\sigma)\log\cZ(\sigma)]
    =F(\sigma).
  \end{aligned}
  \label{eq:gradient-expectation-identity}
\end{equation}
which proves \eqref{eq:appendix-channel-gradient}.  Finally, the Bregman
expansion and \eqref{eq:appendix-relative-smoothness} prove
\begin{equation}
  \begin{aligned}
    F(\rho)&\le Q_\sigma(\rho),\\
    F(\sigma)&=Q_\sigma(\sigma).
  \end{aligned}
  \label{eq:appendix-majorization-property}
\end{equation}
\end{proof}

\subsection{Affine information projection Pythagorean identity}
\label{app:affine-pythagorean}

\begin{prop}[Affine information projection Pythagorean identity]
  \label{prop:affine-pythagorean}
  Let $\chi$ be a positive-definite density operator on $\mathcal H$, and suppose
  a finite multiplier $\lambda^*\in\mathbb R^r$ makes the
  Gibbs state $\rho^*_\chi\coloneqq \rho_{\chi,\lambda^*}$ satisfy all constraints in
  \eqref{eq:feasible-set}.  Then, for every
  $\tau\in\mathcal F$, including singular $\tau$,
  \begin{equation}
    D(\tau\Vert\chi)
    =D(\tau\Vert\rho^*_\chi)
     +D(\rho^*_\chi\Vert\chi).
    \label{eq:affine-pythagorean-appendix}
  \end{equation}
  In particular, $\rho^*_\chi$ is the unique minimizer of
  $D(\rho\Vert\chi)$ over $\mathcal F$.
\end{prop}

\begin{proof}
For every feasible $\tau$ and every finite $\lambda$,
\begin{equation}
  \log\rho_{\chi,\lambda}
  =\log\chi-\lambda\cdot M
    -\log Z_\chi(\lambda)\mathbbm1_{\mathcal H},
  \label{eq:gibbs-log-identity}
\end{equation}
and hence feasibility gives
\begin{equation}
  D(\tau\Vert\rho_{\chi,\lambda})
  =D(\tau\Vert\chi)+g_\chi(\lambda).
  \label{eq:gibbs-variational-identity}
\end{equation}
Evaluate \eqref{eq:gibbs-variational-identity} first with arbitrary
$\tau\in\mathcal F$ and then with $\tau=\rho^*_\chi$, both at
$\lambda=\lambda^*$, and subtract.  This proves
\eqref{eq:affine-pythagorean-appendix}.  Both reference states are
positive, so no support qualification on the possibly singular first
argument is needed.  Nonnegativity and faithfulness of relative entropy prove
uniqueness.  The general matrix projection framework is given by
Ji~\cite{ji2022matrix}, Lemma~3.10 and Theorem~3.1; the calculation above is
the direct proof for the present affine slice.
\end{proof}

\subsection{Moving-reference and mirror three-point identities}
\label{app:mirror-identities}

\begin{prop}[Moving-reference and mirror three-point identities]
  \label{prop:mirror-identities}
  Let $\sigma\in\mathcal F$ be positive, define $\chi(\sigma)$ by
  \eqref{eq:moving-reference}, and let
  $\eta\coloneqq T(\sigma)$ be its exact affine information projection.  For every
  density operator $\tau$ on $\mathcal H$, including singular $\tau$,
  \begin{equation}
    D(\tau\Vert\chi(\sigma))
    =D(\tau\Vert\sigma)+\tr(G_\sigma\tau)+\log Z(\sigma).
    \label{eq:moving-reference-identity}
  \end{equation}
  If $\tau\in\mathcal F$, then
  \begin{equation}
    D(\tau\Vert\sigma)-D(\tau\Vert\eta)
    =D(\eta\Vert\sigma)+\tr[G_\sigma(\eta-\tau)].
    \label{eq:mirror-three-point-identity}
  \end{equation}
  Taking $\tau=\sigma$ gives
  \begin{equation}
    \tr[G_\sigma(\sigma-\eta)]
    =D(\sigma\Vert\eta)+D(\eta\Vert\sigma).
    \label{eq:symmetrized-mirror-step}
  \end{equation}
\end{prop}

\begin{proof}
Equation~\eqref{eq:moving-reference} gives
\begin{equation}
  \log\chi(\sigma)
  =\log\sigma-G_\sigma
    -\log Z(\sigma)\mathbbm1_{\mathcal H}.
  \label{eq:moving-reference-logarithm}
\end{equation}
Substitution into relative entropy proves
\eqref{eq:moving-reference-identity}.  Using
\eqref{eq:gradient-expectation-identity} and the definition of
$Q_\sigma$ also gives
\begin{equation}
  D(\rho\Vert\chi(\sigma))
  =Q_\sigma(\rho)+\log Z(\sigma).
  \label{eq:appendix-projection-majorizer-equivalence}
\end{equation}
The additive term is independent of $\rho$, so the affine information
projection is equivalently the entropy-mirror subproblem defined by
$Q_\sigma$ in \eqref{eq:tangent-majorizer}.

Apply Proposition~\ref{prop:affine-pythagorean} to $\eta=T(\sigma)$ and
comparison state $\tau\in\mathcal F$.  Expand both moving-reference terms by
\eqref{eq:moving-reference-identity}; the normalization constant cancels
and gives \eqref{eq:mirror-three-point-identity}.  The specialization
$\tau=\sigma$ yields \eqref{eq:symmetrized-mirror-step}.
\end{proof}

\section{Affine Gibbs projection}
\label{app:gibbs-projection-proof}

Each outer iteration is well defined only if its affine Gibbs subproblem has a
unique feasible solution at a finite multiplier.  This appendix proves that
claim from the Slater and constraint-independence assumptions, thereby
justifying the exact update map used throughout the convergence analysis.
Fix an arbitrary positive-definite density operator $\chi$ on $\mathcal H$;
the Gibbs family and dual objective associated with this reference are those
in \eqref{eq:gibbs-dual-objective}.

\begin{lem}[Optimality of affine Gibbs states]
  \label{lem:gibbs-projection}
  For every $\chi\succ0$, the function $g_\chi$ has a unique finite minimizer
  $\lambda_\chi^*$.  The state $\rho^*_\chi\coloneqq 
  \rho_{\chi,\lambda_\chi^*}$ is positive definite, matches every affine
  moment, and is the unique information projection
  \begin{equation}
    \rho^*_\chi
    =\operatorname*{argmin}_{\rho\in\mathcal F}D(\rho\Vert\chi).
    \label{eq:appendix-gibbs-projection-statement}
  \end{equation}
  If $r=0$, this projection is $\chi$ itself.
\end{lem}
\noindent This is the formal version of the Gibbs structure result stated in
Sec.~\ref{sec:affine-gibbs-projection}.

\begin{proof}[Proof of Lemma~\ref{lem:gibbs-projection}]
If $r=0$, there are no nontrivial affine moments.  The normalized reference
$\chi$ is feasible, and $D(\rho\Vert\chi)\ge0$ with equality only at
$\rho=\chi$.  Assume henceforth that $n\ge1$.

For a unit vector $v\in\mathbb R^r$, set $A_v\coloneqq \sum_iv_iM_i$.  Feasibility
of a Slater state $\rho_\mathrm{sl} \succ 0$ gives
\begin{equation}
  v\cdot m-\lambda_{\min}(A_v)
  =\tr\!\left[
    \rho_{\mathrm{sl}}(A_v-\lambda_{\min}(A_v)\mathbbm1)
  \right].
  \label{eq:slater-directional-margin}
\end{equation}
The operator in parentheses is positive semidefinite and nonzero: otherwise
$A_v$ would be a scalar multiple of the identity, contrary to independence
of $\{\mathbbm1,M_1,\ldots,M_n\}$.  Since
$\rho_{\mathrm{sl}}\succ0$, continuity and compactness of the unit sphere
give the uniform margin
\begin{equation}
  \delta_{\mathrm{sl}}
  \coloneqq \min_{\norm{v}_2=1}
    \bigl(v\cdot m-\lambda_{\min}(A_v)\bigr)>0.
  \label{eq:uniform-slater-margin}
\end{equation}

For nonzero $\lambda$, write $\lambda=tv$, where
$t=\norm{\lambda}_2$ and $\norm{v}_2=1$.  The Rayleigh principle gives
\begin{equation}
  \log\tr\exp(\log\chi-tA_v)
  \ge\lambda_{\min}(\log\chi)-t\lambda_{\min}(A_v),
  \label{eq:log-partition-rayleigh-bound}
\end{equation}
and hence
\begin{equation}
  g_\chi(\lambda)
  \ge\lambda_{\min}(\log\chi)+t\delta_{\mathrm{sl}}.
  \label{eq:dual-coercivity}
\end{equation}
Thus $g_\chi$ is coercive and attains a minimum at a finite multiplier.

Duhamel differentiation of the exponential gives, for $i=1,\ldots,r$,
\begin{equation}
  \bigl(\nabla g_\chi(\lambda)\bigr)_i
  =m_i-\tr(\rho_{\chi,\lambda}M_i).
  \label{eq:appendix-dual-gradient}
\end{equation}
Proposition~\ref{prop:inner-bkm-susceptibility} in
Appendix~\ref{app:bkm-geometry} identifies the Hessian as a BKM covariance.
That covariance is positive definite under the constraint-independence
assumption, so $g_\chi$ is strictly convex and its finite minimizer
$\lambda_\chi^*$ is unique.

Stationarity is exact moment matching for $i=1,\ldots,r$:
\begin{equation}
  \tr(\rho^*_\chi M_i)=m_i.
  \label{eq:exact-moment-matching}
\end{equation}
Together with normalization, these are precisely the affine constraints.
The normalized exponential is positive definite.  Therefore the assumption
of Proposition~\ref{prop:affine-pythagorean} hold with
$\lambda^*=\lambda_\chi^*$, and that proposition proves
\begin{equation}
  \rho^*_\chi
  =\operatorname*{argmin}_{\rho\in\mathcal F}D(\rho\Vert\chi).
  \label{eq:appendix-unique-information projection}
\end{equation}
It also gives
\begin{equation}
  g_\chi(\lambda_\chi^*)
  =-D(\rho^*_\chi\Vert\chi),
  \label{eq:optimal-dual-value}
\end{equation}
by taking $\tau=\rho^*_\chi$ in
\eqref{eq:gibbs-variational-identity}.  This completes the proof.
\end{proof}

\section{Outer certificate}
\label{app:certificate-proof}

A feasible candidate gives an upper bound, but a security calculation also
needs a lower bound on the unknown optimum.  This appendix derives that bound
from a supporting hyperplane and proves its tightness at a positive optimizer
through the affine optimality conditions.

For a positive-definite density operator $\omega$ on $\mathcal H$ and
$\alpha\in\mathbb R^n$, define the certificate pencil and its value by
\begin{equation}
  \begin{aligned}
    C(\omega,\alpha)&\coloneqq G_\omega+\sum_{i=1}^n\alpha_iM_i,\\
    L(\omega,\alpha)&\coloneqq \lambda_{\min}(C(\omega,\alpha))-\alpha\cdot m.
  \end{aligned}
  \label{eq:appendix-certificate-definition}
\end{equation}
\begin{prop}[Exact candidate--certificate bracket]
  \label{prop:candidate-certificate}
  For every $\omega\succ0$, every $\alpha\in\mathbb R^n$, and every
  feasible candidate $\widehat\rho$,
  \begin{equation}
    L(\omega,\alpha)\le F^*\le F(\widehat\rho).
    \label{eq:appendix-candidate-certificate-bracket}
  \end{equation}
  If $z\in\mathcal S^*$ is positive definite, there is an
  $\alpha^*\in\mathbb R^r$ such that $L(z,\alpha^*)=F^*$.
\end{prop}

This proposition is the formal certificate result summarized in
Sec.~\ref{sec:outer-certificate}.

\begin{proof}
Convexity, differentiability, and
\eqref{eq:gradient-expectation-identity} give, for every density operator
$\rho$ on $\mathcal H$,
\begin{equation}
  F(\rho)
  \ge F(\omega)+\tr[G_\omega(\rho-\omega)]
  =\tr(G_\omega\rho).
  \label{eq:appendix-supporting-linearization}
\end{equation}
If
$\rho\in\mathcal F$ and $\alpha\in\mathbb R^n$, then
\begin{equation}
  \tr(G_\omega\rho)
  =\tr\!\left[\left(G_\omega+\sum_i\alpha_iM_i\right)\rho\right]
   -\alpha\cdot m.
  \label{eq:certificate-affine-shift}
\end{equation}
The expectation of a Hermitian matrix in a density operator is at least its
smallest eigenvalue.  This proves the lower bound in
\eqref{eq:appendix-candidate-certificate-bracket}; feasibility of
$\widehat\rho$ proves the upper bound.

For tightness, let $z\in\mathcal S^*$ with $z\succ0$.  Since $z$ lies in the relative
interior of the positive cone within the affine slice, constrained
stationarity and independence of the constraints imply the existence of
$\alpha^*\in\mathbb R^r$ and $\nu\in\mathbb R$ such that
\begin{equation}
  G_z+\sum_i\alpha_i^*M_i=\nu\mathbbm1_{\mathcal H}.
  \label{eq:positive-optimizer-kkt}
\end{equation}
Taking the expectation in $z$ and using feasibility and
\eqref{eq:gradient-expectation-identity} gives
\begin{equation}
  \nu=F^*+\alpha^*\cdot m.
  \label{eq:certificate-kkt-scalar}
\end{equation}
The least eigenvalue of the certificate pencil is $\nu$, and hence
\begin{equation}
  L(z,\alpha^*)=\nu-\alpha^*\cdot m=F^*.
  \label{eq:certificate-tightness}
\end{equation}
The logarithmic linearization was used only at the positive state $\omega$;
the feasible upper candidate may be singular.
\end{proof}

\section{Global convergence}
\label{app:global-convergence-proof}

Objective descent alone does not ensure that repeated exact Gibbs projections
approach the solution of the original problem.  This appendix derives the
telescoping estimate that yields the last-iterate $O(1/n)$ objective rate and
convergence of the states to the optimizer set.  It is the formal counterpart
of Sec.~\ref{sec:outer-convergence}.

\begin{thm}[Global convergence of the exact outer iteration]
  \label{thm:global-convergence}
  For every $\rho^*\in\mathcal S^*$ and every integer $n\ge1$,
  \begin{equation}
    0\le F(\rho_n)-F^*
    \le\frac{D(\rho^*\Vert\rho_0)}{n}
    \le\frac{-\log\lambda_{\min}(\rho_0)}{n}.
    \label{eq:appendix-global-objective-rate}
  \end{equation}
  Moreover,
  \begin{equation}
    \min_{z\in\mathcal S^*}\norm{\rho_n-z}_1\longrightarrow0.
    \label{eq:appendix-optimizer-set-convergence}
  \end{equation}
  If $\mathcal S^*$ is a singleton, the complete state sequence converges to
  its element.  If the compressed preprocessing is identified isometrically
  with the input, so that $\cG=\operatorname{id}$, the first estimate sharpens
  to
  \begin{equation}
    F(\rho_n)-F^*
    \le\frac{D(\cZ(\rho^*)\Vert\cZ(\rho_0))}{n}.
    \label{eq:appendix-isometric-objective-rate}
  \end{equation}
\end{thm}

\begin{proof}[Proof of Theorem~\ref{thm:global-convergence}]
Fix a positive feasible state $\rho$ and put $\eta\coloneqq T(\rho)$.  By
Lemma~\ref{lem:gibbs-projection}, $\eta$ is positive and feasible.
Equation~\eqref{eq:symmetrized-mirror-step} and the Bregman expansion of $F$
give the exact descent identity
\begin{equation}
  F(\rho)-F(\eta)
  =D(\rho\Vert\eta)+D(\eta\Vert\rho)
   -\mathfrak B_F(\eta,\rho).
  \label{eq:exact-descent-identity}
\end{equation}
The last two terms have the decomposition
\begin{align}
  D(\eta\Vert\rho)-\mathfrak B_F(\eta,\rho)
  &=D(\eta\Vert\rho)
    -D(\cG(\eta)\Vert\cG(\rho))
  \notag\\
  &\quad+D(\cZ(\eta)\Vert\cZ(\rho))\ge0.
  \label{eq:descent-nonnegative-decomposition}
\end{align}
Data processing under $\cG$ makes the difference nonnegative, and the last
relative entropy is nonnegative.  Hence $F(\rho_k)$ is nonincreasing.

Fix $\rho^*\in\mathcal S^*$ and set
$\delta_k\coloneqq F(\rho_k)-F^*$.  Subtracting the Bregman expansions of $F(\eta)$
and $F(\rho^*)$ about the same positive base $\rho$ gives
\begin{equation}
  \tr[G_\rho(\eta-\rho^*)]
  =F(\eta)-F^*-\mathfrak B_F(\eta,\rho)
   +\mathfrak B_F(\rho^*,\rho).
  \label{eq:optimizer-bregman-comparison}
\end{equation}
Insert this into the three-point identity
\eqref{eq:mirror-three-point-identity} with
$\tau=\rho^*$.  At step $k$,
\begin{align}
  &D(\rho^*\Vert\rho_k)
    -D(\rho^*\Vert\rho_{k+1})
  \notag\\
  &\quad=\delta_{k+1}
   +D(\rho_{k+1}\Vert\rho_k)
  \notag\\
  &\qquad-\mathfrak B_F(\rho_{k+1},\rho_k)
    +\mathfrak B_F(\rho^*,\rho_k).
  \label{eq:optimizer-comparison-identity}
\end{align}
All terms after $\delta_{k+1}$ are nonnegative by
\eqref{eq:appendix-relative-smoothness}.  The identity remains finite when
$\rho^*$ is singular because every second argument and its channel images are
positive.  Discarding the nonnegative terms proves
\eqref{eq:one-step-telescoping-bound}; summing gives
\begin{equation}
  \sum_{j=1}^n\delta_j\le D(\rho^*\Vert\rho_0).
  \label{eq:gap-telescoping-bound}
\end{equation}
The gaps are nonincreasing, so
$n\delta_n\le\sum_{j=1}^n\delta_j$.  This proves the first rate in
\eqref{eq:appendix-global-objective-rate}.  Since $\rho_0\succ0$,
\begin{equation}
  \begin{aligned}
    D(\rho^*\Vert\rho_0)
    &=-S(\rho^*)
      -\tr(\rho^*\log\rho_0)\\
    &\le-\log\lambda_{\min}(\rho_0).
  \end{aligned}
  \label{eq:initial-relative-entropy-bound}
\end{equation}
which proves its second inequality.

For state convergence, suppose that some subsequence remained at trace
distance at least $\varepsilon>0$ from $\mathcal S^*$.  Compactness of
$\mathcal F$ gives a further subsequence converging to
$\overline\rho\in\mathcal F$.
Equation~\eqref{eq:appendix-global-objective-rate} and
continuity imply $F(\overline\rho)=F^*$, so
$\overline\rho\in\mathcal S^*$, a contradiction.  This proves
\eqref{eq:appendix-optimizer-set-convergence}.  If the optimizer set is a
singleton, state-to-set convergence is convergence of the complete sequence.

For the isometric refinement, identify the preprocessing output support with
the input, so $\cG$ is the identity.  For every density operator $\tau$ on
$\mathcal H$ and positive $\rho\in\mathcal F$, direct substitution into
\eqref{eq:moving-reference-identity} gives
\begin{equation}
  D(\tau\Vert\chi(\rho))
  =F(\tau)+D(\cZ(\tau)\Vert\cZ(\rho))+\log Z(\rho).
  \label{eq:isometric-moving-reference}
\end{equation}
Apply affine Pythagoras with $\tau=\rho^*$ and
$\eta=T(\rho)$.  After
cancelling constants,
\begin{align}
  D(\cZ(\rho^*)\Vert\cZ(\rho))
  &=F(\eta)-F^*+D(\rho^*\Vert\eta)
  \notag\\
  &\quad+D(\cZ(\eta)\Vert\cZ(\rho)).
  \label{eq:isometric-optimizer-comparison}
\end{align}
Data processing under $\cZ$ gives
$D(\rho^*\Vert\eta)
\ge D(\cZ(\rho^*)\Vert\cZ(\eta))$.  Thus
\begin{equation}
  D(\cZ(\rho^*)\Vert\cZ(\rho_k))
  -D(\cZ(\rho^*)\Vert\cZ(\rho_{k+1}))
  \ge\delta_{k+1}.
  \label{eq:isometric-telescoping-potential}
\end{equation}
Telescoping and monotonicity give
\eqref{eq:appendix-isometric-objective-rate}.
\end{proof}

\section{Interior-optimizer selection}
\label{app:interior-selection-proof}

Global convergence to an optimizer set does not by itself select one optimizer
or keep the iterates uniformly away from the boundary.  Assuming a positive
optimizer exists, this appendix proves both properties and derives the spectral
floors needed for the subsequent local and inner-solver analyses.

Assume that $\rho^\sharp\in\mathcal S^*$ is positive definite, and define
\begin{equation}
  \begin{aligned}
    a^\sharp&\coloneqq \lambda_{\min}(\rho^\sharp),\\
    a&\coloneqq \exp\!\left[-\frac{
      D(\rho^\sharp\Vert\rho_0)+S(\rho^\sharp)}{a^\sharp}\right].
  \end{aligned}
  \label{eq:appendix-interior-floor-constants}
\end{equation}
For the channel-output bounds, define
\begin{equation}
  \begin{aligned}
    c_G&\coloneqq \lambda_{\min}(\cG(\mathbbm1_{\mathcal H})),\\
    C_G&\coloneqq \lambda_{\max}(\cG(\mathbbm1_{\mathcal H})),\\
    c_Z&\coloneqq \lambda_{\min}(\cZ(\mathbbm1_{\mathcal H})),\\
    C_Z&\coloneqq \lambda_{\max}(\cZ(\mathbbm1_{\mathcal H})).
  \end{aligned}
  \label{eq:appendix-channel-spectral-constants}
\end{equation}
The support compression makes all four constants strictly positive.

\begin{cor}[Selection of an interior optimizer]
  \label{cor:interior-selection}
  Suppose $\mathcal S^*$ contains $\rho^\sharp\succ0$, and define
  $a^\sharp,a,c_G,C_G,c_Z,C_Z$ as above.  Then
  $\rho_k\succeq a\mathbbm1$ for every $k$, and the complete sequence
  converges to a positive-definite optimizer $\rho^*$ selected by the
  trajectory from $\rho_0$.  The exact multiplier sequence also converges,
  and
  \begin{equation}
    \chi(\rho_k)
    \succeq\frac{a^3c_Gc_Z}{dC_GC_Z}\mathbbm1.
    \label{eq:appendix-selection-reference-floor}
  \end{equation}
  For isometric preprocessing the right-hand coefficient improves to
  $ac_Z/(dC_Z)$; for direct pinching,
  $\chi(\rho_k)=\Phi(\rho_k)\succeq a\mathbbm1$.
\end{cor}

This is the formal selection result summarized in
Sec.~\ref{sec:outer-convergence}.

\begin{proof}[Proof of Corollary~\ref{cor:interior-selection}]
Apply \eqref{eq:optimizer-comparison-identity} with
$\rho^*=\rho^\sharp$.  Its nonnegative right-hand
terms imply, for every $k$,
\begin{equation}
  D(\rho^\sharp\Vert\rho_k)
  \le D(\rho^\sharp\Vert\rho_0).
  \label{eq:interior-fejer-bound}
\end{equation}
Every density operator satisfies
$\rho_k\preceq\mathbbm1_{\mathcal H}$, so
$-\log\rho_k\succeq0$.  Since
$\rho^\sharp\succeq a^\sharp\mathbbm1_{\mathcal H}$,
\begin{align}
  D(\rho^\sharp\Vert\rho_0)+S(\rho^\sharp)
  &\ge\tr[\rho^\sharp(-\log\rho_k)]
  \notag\\
  &\ge a^\sharp\tr(-\log\rho_k)
  \notag\\
  &\ge a^\sharp[-\log\lambda_{\min}(\rho_k)].
  \label{eq:iterate-floor-log-bound}
\end{align}
Rearranging and exponentiating gives
$\rho_k\succeq a\mathbbm1_{\mathcal H}$ with $a$ from
\eqref{eq:appendix-interior-floor-constants}.

Appendix~\ref{app:global-convergence-proof} and compactness provide a cluster
point $\rho^*\in\mathcal S^*$.  This is the optimizer selected by the exact
trajectory, and the common floor makes it positive.  Using
\eqref{eq:optimizer-comparison-identity} shows that
$D(\rho^*\Vert\rho_k)$ is nonincreasing; along a subsequence converging
to $\rho^*$ it tends to zero.  Hence
\begin{equation}
  D(\rho^*\Vert\rho_k)\longrightarrow0.
  \label{eq:selected-relative-entropy-convergence}
\end{equation}
Natural-log quantum Pinsker
\cite[Theorem~5.38]{Watrous2018} yields
\begin{equation}
  \norm{\rho_k-\rho^*}_1^2
  \le2D(\rho^*\Vert\rho_k)\longrightarrow0.
  \label{eq:selected-trace-convergence}
\end{equation}

For $r\ge1$, the exact multiplier solves
$\nabla g_{\chi(\rho)}(\lambda)=0$.  Its multiplier Jacobian is the
positive-definite Hessian characterized by
\eqref{eq:directional-dual-hessian}.  The implicit
function theorem makes the unique multiplier a continuously differentiable
function of the positive reference near $\chi(\rho^*)$.  Proposition
\ref{prop:channel-bregman} and
\eqref{eq:moving-reference} make
$\rho\mapsto\chi(\rho)$ smooth in the common positive neighborhood.  Thus
$\rho_k\to\rho^*$ implies convergence of the exact multipliers.  For
$r=0$ the claim is vacuous.

It remains to prove the moving-reference floor
\begin{equation}
  \chi(\rho_k)
  \succeq\frac{a^3c_Gc_Z}{dC_GC_Z}\mathbbm1_{\mathcal H}.
  \label{eq:appendix-general-reference-floor}
\end{equation}
From
$a\mathbbm1_{\mathcal H}\preceq\rho_k\preceq\mathbbm1_{\mathcal H}$ and
positivity,
\begin{align}
  ac_G\mathbbm1_{\mathcal K_G}
  &\preceq\cG(\rho_k)\preceq C_G\mathbbm1_{\mathcal K_G},
  \label{eq:preprocessing-iterate-bounds}\\
  ac_Z\mathbbm1_{\mathcal K_Z}
  &\preceq\cZ(\rho_k)\preceq C_Z\mathbbm1_{\mathcal K_Z}.
  \label{eq:output-iterate-bounds}
\end{align}
The adjoints are positive and unital.  Applying the logarithm and the
adjoints gives
\begin{equation}
  \log\!\left(\frac{a^2c_Z}{C_G}\right)\mathbbm1_{\mathcal H}
  \preceq H(\rho_k)
  \preceq
  \log\!\left(\frac{C_Z}{ac_G}\right)\mathbbm1_{\mathcal H}.
  \label{eq:mirror-hamiltonian-bounds}
\end{equation}
The smallest eigenvalue of $e^{H(\rho_k)}$ is at least $a^2c_Z/C_G$, while
its trace is at most $dC_Z/(ac_G)$.  Their quotient proves
\eqref{eq:appendix-general-reference-floor}.

For isometric preprocessing, the terms
$\log\rho-\cG^*(\log\cG(\rho))$ cancel under the support identification.
The same argument using only $\cZ$ gives
$\chi(\rho_k)\succeq ac_Z/(dC_Z)\mathbbm1_{\mathcal H}$.  For direct pinching,
self-adjointness and functional calculus give
$H(\rho_k)=\log\Phi(\rho_k)$ and
$\chi(\rho_k)=\Phi(\rho_k)\succeq a\mathbbm1_{\mathcal H}$.
\end{proof}

\section{BKM geometry and local convergence}
\label{app:bkm-geometry}
\label{app:geometric-convergence-proof}

The BKM geometry enters the analysis in two places: as the susceptibility of
the inner Gibbs family and as the local metric generated by relative entropy
around a positive optimizer.  This appendix develops both roles.  It contains
the formulas omitted from Secs.~\ref{sec:affine-gibbs-projection} and
\ref{sec:geometric-convergence}, and then proves the formal local convergence
result.

\subsection{Gibbs susceptibility}
\label{app:inner-bkm-susceptibility}

For a positive state $\rho$, define the Kubo--Mori operator
\begin{equation}
  \cJ_\rho(X)\coloneqq \int_0^1\rho^uX\rho^{1-u}\,\mathrm du.
  \label{eq:kubo-mori-operator}
\end{equation}
For the Gibbs family in \eqref{eq:gibbs-dual-objective}, set
\begin{equation}
  \mu_i(\lambda)\coloneqq \tr(\rho_{\chi,\lambda}M_i),
  \qquad
  \overline M_i(\lambda)\coloneqq M_i-\mu_i(\lambda)\mathbbm1.
  \label{eq:appendix-centered-observables}
\end{equation}

\begin{prop}[BKM susceptibility of the Gibbs family]
  \label{prop:inner-bkm-susceptibility}
  The Hessian of $g_\chi$ is
  \begin{equation}
    \bigl(\nabla^2g_\chi(\lambda)\bigr)_{ij}
    =\tr\!\left[
      \overline M_i(\lambda)\,
      \cJ_{\rho_{\chi,\lambda}}(\overline M_j(\lambda))
    \right].
    \label{eq:dual-hessian}
  \end{equation}
  It is positive definite for every finite $\lambda$ when
  $\{\mathbbm1,M_1,\ldots,M_r\}$ is real-linearly independent.  If
  $\log\chi-\lambda\cdot M=\sum_a\kappa_a\ket{a}\!\bra{a}$ and
  $w_a\coloneqq e^{\kappa_a}/\sum_be^{\kappa_b}$, then
  \begin{equation}
    \begin{aligned}
      \bigl(\nabla^2g_\chi(\lambda)\bigr)_{ij}
      ={}&\sum_{a,b}\ell(w_a,w_b)
        \bra{a}\overline M_i(\lambda)\ket{b}\\
        &\qquad\times
        \bra{b}\overline M_j(\lambda)\ket{a},
    \end{aligned}
    \label{eq:spectral-dual-hessian}
  \end{equation}
  where $\ell(x,y)\coloneqq (x-y)/(\log x-\log y)$ for $x\ne y$ and
  $\ell(x,x)\coloneqq x$.
\end{prop}

\begin{proof}
Duhamel differentiation of the matrix exponential gives
$\partial_{\lambda_i}\log Z_\chi=-\mu_i$.  Differentiating once more and
centering the observables removes the normalization derivative, yielding
\eqref{eq:dual-hessian}.  In a direction $v\in\mathbb R^r$, with
$A_v\coloneqq \sum_iv_iM_i$ and
$\mu_v\coloneqq \tr(\rho_{\chi,\lambda}A_v)$, this reads
\begin{equation}
  v^\mathsf T\nabla^2g_\chi(\lambda)v
  =\tr\!\left[
    (A_v-\mu_v\mathbbm1)\,
    \cJ_{\rho_{\chi,\lambda}}(A_v-\mu_v\mathbbm1)
  \right].
  \label{eq:directional-dual-hessian}
\end{equation}
Every coefficient of $\cJ_{\rho_{\chi,\lambda}}$ in an eigenbasis of the
positive Gibbs state is a strictly positive logarithmic mean.  The quadratic
form can therefore vanish only if $A_v=\mu_v\mathbbm1$, which constraint
independence forces to have $v=0$.  This proves positive definiteness.
Evaluating the same logarithmic means in the eigenbasis of the Gibbs
Hamiltonian gives \eqref{eq:spectral-dual-hessian}.
\end{proof}

\subsection{Local geometry of the exact outer map}
\label{app:outer-bkm-geometry}

The global $O(1/n)$ bound does not capture the faster local behavior near the
selected positive optimizer.  We now linearize the exact update in BKM
geometry, identify its transverse contraction, and obtain geometric state
convergence together with a quadratic objective rate.

Assume the assumption of Appendix~\ref{app:interior-selection-proof}, and let
$\rho^*\succ0$ be the selected optimizer to which the exact sequence
$\{\rho_k\}_{k\ge0}$ converges.  The real feasible tangent space is
\begin{equation}
  \mathcal V
  \coloneqq \{X\in\mathcal B(\mathcal H):X=X^\dagger,
       \tr X=0,\ \tr(XM_i)=0,\ i=1,\ldots,r\}.
  \label{eq:appendix-feasible-tangent-space}
\end{equation}
For a positive state $\tau$ on any of the relevant Hilbert spaces, define the
Fr\'echet derivative of the logarithm by
\begin{equation}
  \Omega_\tau(X)\coloneqq \mathrm D\log_\tau[X]
  \coloneqq \int_0^\infty
    (\tau+t\mathbbm1)^{-1}X(\tau+t\mathbbm1)^{-1}\,\mathrm dt.
  \label{eq:appendix-logarithmic-derivative}
\end{equation}
Set
\begin{equation}
  \begin{aligned}
    \rho_G^*&\coloneqq \cG(\rho^*),\\
    \rho_Z^*&\coloneqq \cZ(\rho^*).
  \end{aligned}
  \label{eq:appendix-optimizer-channel-images}
\end{equation}
On $\mathcal V$, define the real symmetric
bilinear forms
\begin{align}
  \mathcal A(X,Y)
  &\coloneqq \tr[X\Omega_{\rho^*}(Y)],
  \label{eq:appendix-input-bkm-form}\\
  \mathcal B_G(X,Y)
  &\coloneqq \tr[
      \cG(X)\Omega_{\rho_G^*}(\cG(Y))],
  \label{eq:appendix-preprocessing-bkm-form}\\
  \mathcal B_Z(X,Y)
  &\coloneqq \tr[
      \cZ(X)\Omega_{\rho_Z^*}(\cZ(Y))].
  \label{eq:appendix-output-bkm-form}
\end{align}
Let
\begin{equation}
  \begin{aligned}
    \nabla^2F(\rho^*)&\coloneqq \mathcal B_G-\mathcal B_Z,\\
    \mathcal M&\coloneqq \mathcal A-\nabla^2F(\rho^*).
  \end{aligned}
  \label{eq:appendix-flat-transverse-forms}
\end{equation}
Define the bilinear kernel and its $\mathcal A$-orthogonal complement by
\begin{equation}
  \begin{aligned}
    \mathcal K
    &\coloneqq \{X\in\mathcal V:
         \nabla^2F(\rho^*)[X,Y]=0
         \ \text{for every }Y\in\mathcal V\},\\
    \mathcal N&\coloneqq \mathcal K^{\perp_{\mathcal A}}.
  \end{aligned}
  \label{eq:appendix-flat-transverse-spaces}
\end{equation}
The induced norm and distance to a nonempty set $\mathcal C$ are
\begin{equation}
  \begin{aligned}
    \norm{X}_{\mathcal A}&\coloneqq \sqrt{\mathcal A(X,X)},\\
    \operatorname{dist}_{\mathcal A}(\rho,\mathcal C)
    &\coloneqq \inf_{z\in\mathcal C}\norm{\rho-z}_{\mathcal A}.
  \end{aligned}
  \label{eq:appendix-bkm-norm-distance}
\end{equation}
If $\mathcal N\ne\{0\}$, define
\begin{equation}
  q_\perp
  \coloneqq \max_{X\in\mathcal N\setminus\{0\}}
       \frac{\mathcal M(X,X)}{\mathcal A(X,X)},
  \label{eq:appendix-transverse-contraction-factor}
\end{equation}
and set $q_\perp\coloneqq 0$ when $\mathcal N=\{0\}$.
\begin{thm}[Geometric state convergence]
  \label{thm:geometric-convergence}
  The complete optimizer set is
  \begin{equation}
    \mathcal S^*=\mathcal F\cap(\rho^*+\mathcal K),
    \label{eq:appendix-affine-optimizer-set}
  \end{equation}
  and every positive-definite optimizer is a fixed point of $T$.  If
  $\mathcal N\ne\{0\}$, then $0\le q_\perp<1$, and for every
  $\theta\in(q_\perp,1)$ there exist
  $c_\theta,\widetilde c_\theta>0$ and $k_\theta\in\mathbb N$ such that
  \begin{align}
    \norm{\rho_k-\rho^*}_{\mathcal A}
    &\le c_\theta\theta^k,
    \label{eq:appendix-geometric-state-rate}\\
    0\le F(\rho_k)-F^*
    &\le\widetilde c_\theta\theta^{2k}
    \label{eq:appendix-geometric-objective-rate}
  \end{align}
  for all $k\ge k_\theta$.  The state bound remains valid, with a different
  constant, in every fixed norm on the finite-dimensional state space.  If
  $\mathcal N=\{0\}$, every feasible state is optimal and the positive exact
  sequence is stationary.
\end{thm}

This is the formal version of the local result summarized in
Sec.~\ref{sec:geometric-convergence}.

\begin{proof}[Proof of Theorem~\ref{thm:geometric-convergence}]
In an eigenbasis $\tau=\sum_a p_a\ket{a}\!\bra{a}$, the logarithmic derivative
has coefficients
\begin{equation}
  (\Omega_\tau(X))_{ab}
  =\begin{cases}
    X_{ab}/p_a,&p_a=p_b,\\[1mm]
    \dfrac{\log p_a-\log p_b}{p_a-p_b}X_{ab},&p_a\ne p_b.
  \end{cases}
  \label{eq:logarithmic-derivative-spectral}
\end{equation}
Every coefficient is positive.  The Kubo--Mori operator has reciprocal
coefficients
\begin{equation}
  (\cJ_\tau(X))_{ab}
  =\begin{cases}
    p_aX_{ab},&p_a=p_b,\\[1mm]
    \dfrac{p_a-p_b}{\log p_a-\log p_b}X_{ab},&p_a\ne p_b.
  \end{cases}
  \label{eq:kubo-mori-spectral}
\end{equation}
Thus $\cJ_\tau=\Omega_\tau^{-1}$ and $\mathcal A$ is an inner product on
$\mathcal V$.  Twice differentiating $h(\rho)=\tr(\rho\log\rho)$ gives
\begin{equation}
  \nabla^2h(\tau)[X,Y]
  =\tr(X\Omega_\tau(Y)).
  \label{eq:negative-entropy-hessian}
\end{equation}
Consequently, for Hermitian traceless $X$ and real $t\to0$ such that
$\tau+tX\succeq0$,
\begin{equation}
  D(\tau+tX\Vert\tau)
  =\frac{t^2}{2}\tr[X\Omega_\tau(X)]+o(t^2).
  \label{eq:appendix-local-relative-entropy-expansion}
\end{equation}

BKM monotonicity, first under $\cG$ and then under $\Phi$, gives, for every
$X\in\mathcal V$,
\begin{equation}
  0\le\mathcal B_Z(X,X)
  \le\mathcal B_G(X,X)\le\mathcal A(X,X);
  \label{eq:appendix-bkm-monotonicity-chain}
\end{equation}
see Ref.~\cite{Lesniewski1999}, Theorem~2.14.  Differentiating
$F=h\circ\cG-h\circ\cZ$ twice at $\rho^*$ shows that its feasible
Hessian is
\begin{equation}
  \nabla^2F(\rho^*)[X,Y]
  =\mathcal B_G(X,Y)-\mathcal B_Z(X,Y).
  \label{eq:objective-bkm-hessian}
\end{equation}
Consequently $\nabla^2F(\rho^*)$ is positive semidefinite and
\begin{equation}
  0\preceq\mathcal M
  =\mathcal A-\nabla^2F(\rho^*)
  \preceq\mathcal A
  \label{eq:update-form-bounds}
\end{equation}
as quadratic forms.

We first show that positive optimizers are fixed points.  Let $z\succ0$ be an
optimizer.  Differentiable constrained optimality gives, for every
$X\in\mathcal V$,
\begin{equation}
  \tr(G_zX)=0.
  \label{eq:positive-optimizer-stationarity}
\end{equation}
For every $\rho\in\mathcal F$, the moving-reference identity gives
\begin{equation}
  \begin{aligned}
    D(\rho\Vert\chi(z))-D(z\Vert\chi(z))
    &=D(\rho\Vert z)+\tr[G_z(\rho-z)]\\
    &=D(\rho\Vert z)\ge0.
  \end{aligned}
  \label{eq:optimizer-projection-comparison}
\end{equation}
Thus $z$ is the unique affine information projection of $\chi(z)$ and
\begin{equation}
  T(z)=z.
  \label{eq:positive-optimizer-fixed-point}
\end{equation}

At the selected optimizer, stationarity and the channel Bregman identity
give, for every $\rho\in\mathcal F$,
\begin{equation}
  F(\rho)-F^*
  =D(\cG(\rho)\Vert\rho_G^*)
   -D(\cZ(\rho)\Vert\rho_Z^*).
  \label{eq:stationary-channel-gap}
\end{equation}
We characterize equality in this expression.  Define the Petz map for
$\Phi$ at the positive base $\rho_G^*$ by
\begin{equation}
  \mathcal R_\Phi:\mathcal B(\mathcal K_Z)\longrightarrow
    \mathcal B(\mathcal K_G),
  \qquad
  \mathcal R_\Phi(Y)
  \coloneqq (\rho_G^*)^{1/2}
    \Phi^*\!\left(
      (\rho_Z^*)^{-1/2}Y(\rho_Z^*)^{-1/2}
    \right)
    (\rho_G^*)^{1/2}.
  \label{eq:petz-map}
\end{equation}
It is completely positive.  Since $\Phi(\rho_G^*)=\rho_Z^*$, it satisfies
\begin{equation}
  \begin{aligned}
    \tr\mathcal R_\Phi(Y)&=\tr Y,\\
    \mathcal R_\Phi(\rho_Z^*)&=\rho_G^*.
  \end{aligned}
  \label{eq:petz-map-properties}
\end{equation}
Thus $\mathcal R_\Phi$ is CPTP.

For the BKM representing function, let
\begin{equation}
  h_0(x)\coloneqq 
  \begin{cases}
    \dfrac{x-1}{\log x},&x\ne1,\\[1mm]
    1,&x=1.
  \end{cases}
  \label{eq:bkm-representing-function}
\end{equation}
Its reciprocal has the representation
\begin{equation}
  \frac1{h_0(x)}
  =\int_0^\infty\frac{1+t}{x+t}\,\frac{\mathrm dt}{(1+t)^2},
  \label{eq:bkm-reciprocal-representation}
\end{equation}
whose representing measure has full support on $(0,\infty)$.  The adjoint of
$\Phi$ is Schwarz, both base states are positive, and $\cG(X)$ is Hermitian.
Hiai's equality theorem for monotone metrics
\cite[Theorem~4.1(ii) and Remark~4.4]{Hiai2023} therefore gives
\begin{equation}
  \begin{aligned}
    \mathcal B_G(X,X)=\mathcal B_Z(X,X)
    \quad\Longleftrightarrow\quad\\
    \mathcal R_\Phi(\cZ(X))=\cG(X).
  \end{aligned}
  \label{eq:bkm-petz-equality}
\end{equation}
Because $\nabla^2F(\rho^*)$ is positive semidefinite, zero
quadratic value is equivalent to membership in its bilinear kernel
$\mathcal K$.

Let $X\in\mathcal K$ and suppose first that
$\rho\coloneqq \rho^*+X\in\mathcal F$ is positive.  By
\eqref{eq:bkm-petz-equality} and linearity,
\begin{equation}
  \mathcal R_\Phi(\cZ(\rho))=\cG(\rho).
  \label{eq:finite-recovery-equality}
\end{equation}
Data processing under $\Phi$ and then under $\mathcal R_\Phi$ gives
\begin{align}
  D(\cG(\rho)\Vert\rho_G^*)
  &\ge D(\cZ(\rho)\Vert\rho_Z^*)
  \notag\\
  &\ge D\!\left(
    \mathcal R_\Phi(\cZ(\rho))
    \middle\Vert
    \mathcal R_\Phi(\rho_Z^*)
  \right)
  \notag\\
  &=D(\cG(\rho)\Vert\rho_G^*).
  \label{eq:recovery-data-processing-sandwich}
\end{align}
Both inequalities are equalities, and
\eqref{eq:stationary-channel-gap} proves $\rho\in\mathcal S^*$.  If
$\rho$ is singular, apply the same argument to $\rho^*+tX$ for
$0\le t<1$ and use entropy continuity as $t\to1$.  Therefore
\begin{equation}
  \mathcal F\cap(\rho^*+\mathcal K)\subseteq\mathcal S^*.
  \label{eq:flat-slice-contained-in-optimizers}
\end{equation}

Conversely, let $\rho^\flat$ be any optimizer.  Convexity makes the
segment $\rho^*+t(\rho^\flat-\rho^*)$ optimal.  Its second
derivative at $t=0$ is
\begin{equation}
  \nabla^2F(\rho^*)
    [\rho^\flat-\rho^*,\rho^\flat-\rho^*]=0.
  \label{eq:optimizer-segment-zero-curvature}
\end{equation}
Positive semidefiniteness puts $\rho^\flat-\rho^*$ in
$\mathcal K$, proving
\eqref{eq:appendix-affine-optimizer-set}.

We next differentiate the exact update.  Absorb the normalization of
$\chi(\rho)$ into the scalar trace multiplier.  The KKT equation for the
affine projection is
\begin{equation}
  \log T(\rho)-\log\rho+G_\rho
  +\gamma_0(\rho)\mathbbm1_{\mathcal H}
  +\sum_i\gamma_i(\rho)M_i=0,
  \label{eq:update-kkt-equation}
\end{equation}
together with the affine constraints on $T(\rho)$, where
$\gamma_0(\rho),\ldots,\gamma_r(\rho)\in\mathbb R$ are the trace and moment
multipliers.  The constraint-multiplier
Jacobian is the positive-definite BKM covariance proved in
Lemma~\ref{lem:gibbs-projection}.  The implicit-function theorem therefore
makes $T$ and the multipliers continuously differentiable near
$\rho^*$.  For $r=0$, the same conclusion follows directly from the
normalized exponential.

Differentiate \eqref{eq:update-kkt-equation} at the fixed point in a
direction $Y\in\mathcal V$ and pair with $X\in\mathcal V$.  The scalar and
affine multiplier terms vanish, leaving
\begin{equation}
  \begin{aligned}
    \mathcal A(X,\mathrm DT_{\rho^*}[Y])
    &=\mathcal A(X,Y)-\nabla^2F(\rho^*)[X,Y]\\
    &=\mathcal M(X,Y).
  \end{aligned}
  \label{eq:update-differential-riesz}
\end{equation}
Equation~\eqref{eq:update-form-bounds} makes $\mathrm DT_{\rho^*}$ positive,
$\mathcal A$-self-adjoint, and contractive.  Moreover,
\begin{equation}
  \ker(\operatorname{id}_{\mathcal V}-\mathrm DT_{\rho^*})=\mathcal K,
  \label{eq:update-unit-eigenspace}
\end{equation}
because
$\mathcal A(X,(\operatorname{id}_{\mathcal V}-\mathrm DT_{\rho^*})[Y])
=\nabla^2F(\rho^*)[X,Y]$ for every $X,Y\in\mathcal V$.  The
orthogonal complement $\mathcal N$ is invariant under $\mathrm DT_{\rho^*}$.
If
$\mathcal N\ne\{0\}$, finite-dimensional spectral theory gives
\begin{equation}
  \norm{\left.\mathrm DT_{\rho^*}\right|_{\mathcal N}}_{\mathcal A}
  =\max_{X\in\mathcal N,\,X\ne0}
    \frac{\mathcal M(X,X)}{\mathcal A(X,X)}
  =q_\perp<1.
  \label{eq:transverse-linear-contraction}
\end{equation}
The inequality is strict because the unit eigenspace has been removed.  If
$\mathcal N=\{0\}$, then $\mathcal V=\mathcal K$ and
\eqref{eq:appendix-affine-optimizer-set} makes every feasible state
optimal.  Every
positive feasible state is then fixed, so the positive exact trajectory is
stationary.

Assume $\mathcal N\ne\{0\}$.  Decompose each feasible state near
$\rho^*$ as
\begin{equation}
  \rho-\rho^*=K+N,
  \label{eq:local-flat-transverse-decomposition}
\end{equation}
where $K\in\mathcal K$ and $N\in\mathcal N$.  Put $z\coloneqq \rho^*+K$.  In a
sufficiently small positive neighborhood,
$z$ is feasible and \eqref{eq:appendix-affine-optimizer-set} makes it
optimal.
Thus
\begin{equation}
  \operatorname{dist}_{\mathcal A}(\rho,\mathcal S^*)
  =\norm{N}_{\mathcal A}.
  \label{eq:local-optimizer-distance}
\end{equation}
Let $\Pi_{\mathcal N}:\mathcal V\to\mathcal N$ be the
$\mathcal A$-orthogonal projection.
Choose $\theta\in(q_\perp,1)$.  Continuity of $\mathrm DT$ and compactness of
the unit sphere in $\mathcal N$ allow the neighborhood to be reduced so that
\begin{equation}
  \norm{\Pi_{\mathcal N}\mathrm DT_\eta[N]}_{\mathcal A}
  \le\theta\norm{N}_{\mathcal A}
  \label{eq:uniform-local-transverse-bound}
\end{equation}
for every relevant $\eta$ and $N\in\mathcal N$.  Since $T(z)=z$, the
fundamental theorem of calculus along $z+tN$ gives
\begin{equation}
  \Pi_{\mathcal N}(T(\rho)-\rho^*)
  =\int_0^1\Pi_{\mathcal N}\mathrm DT_{z+tN}[N],\mathrm dt.
  \label{eq:update-fundamental-theorem}
\end{equation}
Therefore
\begin{equation}
  \operatorname{dist}_{\mathcal A}(T(\rho),\mathcal S^*)
  \le\theta\operatorname{dist}_{\mathcal A}(\rho,\mathcal S^*).
  \label{eq:local-set-contraction}
\end{equation}

Appendix~\ref{app:interior-selection-proof} gives
$\rho_k\to\rho^*$, so the trajectory eventually remains in this
neighborhood and its optimizer-set distance is $O(\theta^k)$.  Let $z_k$ be
its local $\mathcal A$-projection onto the optimizer set.  A local Lipschitz
constant $L_T$ gives
\begin{equation}
  \norm{\rho_{k+1}-\rho_k}_{\mathcal A}
  \le(L_T+1)\operatorname{dist}_{\mathcal A}(\rho_k,\mathcal S^*).
  \label{eq:geometric-increment-bound}
\end{equation}
These increments are geometrically summable.  Their tail converges to the
already selected limit and proves
\eqref{eq:appendix-geometric-state-rate}, after
enlarging $c_\theta$ to absorb the finite prefix.  Norm equivalence proves
the same rate in every fixed finite-dimensional norm.

Finally, the feasible gradient vanishes at each positive optimizer $z_k$, and
the Hessian of $F$ is bounded on the selected positive neighborhood.  Taylor's
theorem therefore supplies $C<\infty$ such that
\begin{equation}
  0\le F(\rho_k)-F^*
  \le C\operatorname{dist}_{\mathcal A}(\rho_k,\mathcal S^*)^2
  \le\widetilde c_\theta\theta^{2k},
  \label{eq:quadratic-objective-growth-local}
\end{equation}
which proves \eqref{eq:appendix-geometric-objective-rate}.
\end{proof}

\section{Inner Newton convergence}
\label{app:inner-newton-proof}

The outer theory assumes that each affine Gibbs projection is solved exactly,
so the reliability of the inner multiplier solver must be established
separately.  This appendix bounds the relevant dual sublevel set and its
derivatives to prove global damped-Newton convergence and eventual quadratic
convergence, with explicit fixed problem constants.

Fix a positive-definite density operator $\chi$ on $\mathcal H$ and assume
$r\ge1$.  Let $\lambda_\chi^*\coloneqq \operatorname*{argmin}_{\lambda\in\mathbb R^r}
g_\chi(\lambda)$, whose existence and uniqueness were proved in
Appendix~\ref{app:gibbs-projection-proof}.  Choose Armijo parameters
$0<\alpha<1/2$ and $0<\beta<1$ and an arbitrary
$\lambda^{(0)}\in\mathbb R^r$.  At iteration $t$, define
\begin{equation}
  \begin{aligned}
    q_t&\coloneqq \nabla g_\chi(\lambda^{(t)}),\\
    B_t&\coloneqq \nabla^2g_\chi(\lambda^{(t)}),\\
    \Delta\lambda^{(t)}&\coloneqq -B_t^{-1}q_t.
  \end{aligned}
  \label{eq:appendix-newton-direction}
\end{equation}
Let $j_t$ be the smallest nonnegative integer for which
$u_t\coloneqq \beta^{j_t}$ satisfies
\begin{equation}
  g_\chi(\lambda^{(t)}+u_t\Delta\lambda^{(t)})
  \le g_\chi(\lambda^{(t)})
    +\alpha u_t q_t\cdot\Delta\lambda^{(t)},
  \label{eq:appendix-armijo-condition}
\end{equation}
and set
$\lambda^{(t+1)}\coloneqq \lambda^{(t)}+u_t\Delta\lambda^{(t)}$.

\begin{thm}[Convergence of the exact inner solve]
  \label{thm:inner-newton-convergence}
  Let $r\ge1$.  For every fixed $\chi\succ0$ and finite initial multiplier,
  the exact damped Armijo--Newton sequence is well defined and converges to
  $\lambda_\chi^*$.  On its initial dual sublevel set there are constants
  $0<\mu\le L<\infty$ and $M<\infty$ such that
  \begin{equation}
    \mu\mathbbm1_r
    \preceq\nabla^2g_\chi(\lambda)
    \preceq L\mathbbm1_r,
    \label{eq:inner-hessian-bounds}
  \end{equation}
  and the Hessian is $M$-Lipschitz.  After finitely many damped steps, unit
  steps are accepted and
  \begin{equation}
    \norm{\lambda^{(t+1)}-\lambda_\chi^*}_2
    \le\frac{M}{2\mu}
      \norm{\lambda^{(t)}-\lambda_\chi^*}_2^2.
    \label{eq:inner-quadratic-rate}
  \end{equation}
  Moreover,
  \begin{align}
    g_\chi(\lambda)-g_\chi(\lambda_\chi^*)
    &=D(\rho^*_\chi\Vert\rho_{\chi,\lambda}),
    \label{eq:appendix-inner-dual-gap}\\
    \norm{\rho_{\chi,\lambda}-\rho^*_\chi}_1
    &\le\sqrt{2\bigl[g_\chi(\lambda)-g_\chi(\lambda_\chi^*)\bigr]}.
    \label{eq:appendix-inner-state-gap}
  \end{align}
  For $r=0$, no multiplier iteration is present.
\end{thm}

This is the formal result summarized in Sec.~\ref{sec:inner-gibbs-solver}.

\begin{proof}[Proof of Theorem~\ref{thm:inner-newton-convergence}]
The $r=0$ branch was handled in
Appendix~\ref{app:gibbs-projection-proof}.  Retain the uniform Slater margin
$\delta_{\mathrm{sl}}$ from \eqref{eq:uniform-slater-margin}.  For
$v\in\mathbb R^r$, set $A_v\coloneqq \sum_iv_iM_i$ and define
\begin{align}
  \kappa
  &\coloneqq \max_{\norm{v}_2=1}\norm{A_v}_{\mathrm{op}}>0,
  \label{eq:observable-operator-bound}\\
  \gamma_\perp
  &\coloneqq \min_{\norm{v}_2=1}\min_{c\in\mathbb R}
    \norm{A_v-c\mathbbm1_{\mathcal H}}_{\mathrm F}^2>0.
  \label{eq:observable-scalar-separation}
\end{align}
The second inequality follows from independence of
$\{\mathbbm1,M_1,\ldots,M_r\}$ and compactness of the unit sphere.

Let $\lambda^{(0)}$ be the initial multiplier and set
\begin{equation}
  \mathcal L_0
  \coloneqq \{\lambda\in\mathbb R^r:
       g_\chi(\lambda)\le g_\chi(\lambda^{(0)})\}.
  \label{eq:initial-dual-sublevel}
\end{equation}
Put
\begin{align}
  R_0
  &\coloneqq \frac{g_\chi(\lambda^{(0)})-\lambda_{\min}(\log\chi)}
          {\delta_{\mathrm{sl}}},
  \label{eq:dual-sublevel-radius}\\
  p_0
  &\coloneqq \frac1d\exp\!\left[-\operatorname{osc}(\log\chi)-2\kappa R_0\right].
  \label{eq:sublevel-gibbs-floor}
\end{align}
The coercive estimate in \eqref{eq:dual-coercivity} gives
$\norm{\lambda}_2\le R_0$ throughout $\mathcal L_0$.  For
$K_\lambda\coloneqq \log\chi-\lambda\cdot M$,
\begin{equation}
  \operatorname{osc}(K_\lambda)
  \le\operatorname{osc}(\log\chi)+2\kappa R_0.
  \label{eq:sublevel-hamiltonian-oscillation}
\end{equation}
The smallest normalized exponential weight is therefore at least $p_0$.
Hence, for every $\lambda\in\mathcal L_0$,
\begin{equation}
  \rho_{\chi,\lambda}\succeq p_0\mathbbm1_{\mathcal H},
  \label{eq:sublevel-state-floor}
\end{equation}
with $\operatorname{osc}(X)\coloneqq \lambda_{\max}(X)-\lambda_{\min}(X)$.
For a unit $v$, set
$\mu_v(\lambda)\coloneqq \tr(\rho_{\chi,\lambda}A_v)$.  Then
\eqref{eq:directional-dual-hessian} and the spectral formula for $\cJ$
give
\begin{equation}
  v^\mathsf T\nabla^2g_\chi(\lambda)v
  \ge p_0\norm{A_v-\mu_v\mathbbm1_{\mathcal H}}_{\mathrm F}^2
  \ge p_0\gamma_\perp.
  \label{eq:inner-hessian-lower-bound}
\end{equation}
Every logarithmic mean is at most the corresponding arithmetic mean, and
hence
\begin{equation}
  v^\mathsf T\nabla^2g_\chi(\lambda)v
  \le\tr[\rho_{\chi,\lambda}
    (A_v-\mu_v\mathbbm1_{\mathcal H})^2]
  \le\norm{A_v}_{\mathrm{op}}^2
  \le\kappa^2.
  \label{eq:inner-hessian-upper-bound}
\end{equation}
Set
\begin{equation}
  \begin{aligned}
    \mu&\coloneqq p_0\gamma_\perp,\\
    L&\coloneqq \kappa^2.
  \end{aligned}
  \label{eq:explicit-inner-curvature-constants}
\end{equation}
Then, throughout $\mathcal L_0$,
\begin{equation}
  \mu\mathbbm1_r
  \preceq\nabla^2g_\chi(\lambda)
  \preceq L\mathbbm1_r.
  \label{eq:appendix-inner-hessian-bounds}
\end{equation}

To control the Hessian variation, apply the Duhamel simplex formula and
Schatten H\"older to the trace-exponential functional.  For $1\le j\le3$,
\begin{equation}
  \left|\mathrm D^j[\tr e^{(\mathord\cdot)}]_K
    [K_1,\ldots,K_j]\right|
  \le (\tr e^K)\prod_{\ell=1}^j\norm{K_\ell}_{\mathrm{op}}.
  \label{eq:partition-derivative-bound}
\end{equation}
Three derivatives of $K\mapsto\log\tr e^K$ produce one third-moment term,
three products of a first and a second moment, and two triple products.  Their
absolute coefficients sum to six.  Each unit multiplier direction changes
the Hamiltonian by an operator of norm at most $\kappa$, so the third
derivative tensor has
Euclidean operator norm at most
\begin{equation}
  M\coloneqq 6\kappa^3.
  \label{eq:explicit-hessian-lipschitz-constant}
\end{equation}
Thus the Hessian is globally $M$-Lipschitz.

Let $q\coloneqq \nabla g_\chi(\lambda)$ and
$B\coloneqq \nabla^2g_\chi(\lambda)$, and set
$\Delta\lambda\coloneqq -B^{-1}q$.  Away from stationarity,
\begin{equation}
  q\cdot\Delta\lambda=-q^\mathsf TB^{-1}q<0,
  \label{eq:newton-descent-direction}
\end{equation}
so a sufficiently short step satisfies Armijo.  Every accepted step decreases
$g_\chi$ and keeps the iterates in the compact set $\mathcal L_0$.
The global upper Hessian bound gives, for $0\le u\le1$,
\begin{equation}
  g_\chi(\lambda+u\Delta\lambda)
  \le g_\chi(\lambda)+u q\cdot\Delta\lambda
     +\frac{L}{2}u^2\norm{\Delta\lambda}_2^2.
  \label{eq:armijo-taylor-upper-bound}
\end{equation}
On $\mathcal L_0$,
\begin{equation}
  -q\cdot\Delta\lambda
  =\Delta\lambda^\mathsf TB\Delta\lambda
  \ge\mu\norm{\Delta\lambda}_2^2.
  \label{eq:newton-decrement-lower-bound}
\end{equation}
Therefore every trial satisfying
$u\le2(1-\alpha)\mu/L$ passes.  Backtracking
returns a step bounded below by
\begin{equation}
  u_{\min}
  \coloneqq \beta\min\!\left\{1,
    \frac{2(1-\alpha)\mu}{L}\right\}>0.
  \label{eq:armijo-step-floor}
\end{equation}
The accepted decrease satisfies
\begin{equation}
  g_\chi(\lambda^{(t)})-g_\chi(\lambda^{(t+1)})
  \ge\frac{\alpha u_{\min}}{L}
       \norm{\nabla g_\chi(\lambda^{(t)})}_2^2.
  \label{eq:armijo-gradient-decrease}
\end{equation}
The objective is bounded below, so the gradients tend to zero.  Every cluster
point is stationary.  Strict convexity leaves only $\lambda_\chi^*$, and
compactness of $\mathcal L_0$ makes the complete sequence converge to it.

Hessian Lipschitz continuity also gives
\begin{equation}
  g_\chi(\lambda+\Delta\lambda)
  \le g_\chi(\lambda)+\frac12q\cdot\Delta\lambda
     +\frac{M}{6}\norm{\Delta\lambda}_2^3.
  \label{eq:full-newton-step-bound}
\end{equation}
Since
$-q\cdot\Delta\lambda\ge\mu\norm{\Delta\lambda}_2^2$
and $\alpha<1/2$, the full step satisfies Armijo whenever
\begin{equation}
  \frac{M}{6}\norm{\Delta\lambda}_2
  \le\left(\frac12-\alpha\right)\mu.
  \label{eq:eventual-full-step-condition}
\end{equation}
The Newton direction tends to zero, so this condition holds after finitely
many steps.

Put $e_t\coloneqq \lambda^{(t)}-\lambda_\chi^*$.  Once unit steps are accepted,
$\nabla g_\chi(\lambda_\chi^*)=0$ gives
\begin{align}
  e_{t+1}
  &=\nabla^2g_\chi(\lambda^{(t)})^{-1}
    \int_0^1\!\bigl[
      \nabla^2g_\chi(\lambda^{(t)})
  \notag\\[-1mm]
  &\hspace{30mm}
      -\nabla^2g_\chi(\lambda_\chi^*+ue_t)
    \bigr]e_t\,\mathrm du.
  \label{eq:newton-error-remainder}
\end{align}
The inverse bound and Hessian Lipschitz constant prove
\begin{equation}
  \norm{\lambda^{(t+1)}-\lambda_\chi^*}_2
  \le\frac{M}{2\mu}
       \norm{\lambda^{(t)}-\lambda_\chi^*}_2^2
  \label{eq:appendix-inner-quadratic-rate}
\end{equation}
once unit steps are accepted.

Finally, apply \eqref{eq:gibbs-variational-identity} with
$\tau=\rho^*_\chi$ and use
\eqref{eq:optimal-dual-value}.  This gives
\begin{equation}
  g_\chi(\lambda)-g_\chi(\lambda_\chi^*)
  =D(\rho^*_\chi\Vert\rho_{\chi,\lambda}).
  \label{eq:appendix-dual-gap-relative-entropy}
\end{equation}
Natural-log quantum Pinsker
\cite[Theorem~5.38]{Watrous2018} therefore gives
\begin{equation}
  \norm{\rho_{\chi,\lambda}-\rho^*_\chi}_1
  \le\sqrt{2\left[g_\chi(\lambda)-g_\chi(\lambda_\chi^*)\right]}.
  \label{eq:appendix-dual-gap-state-bound}
\end{equation}

Under the assumption of Appendix~\ref{app:interior-selection-proof}, the
reference states have a uniform spectral floor.  For the cold start
$\lambda^{(0)}=0$,
$g_\chi(0)=0$, so the constants above can be bounded uniformly over the exact
outer trajectory.  The convergent preceding exact multiplier likewise yields
uniformly bounded warm-start sublevels.  This is uniformity for one fixed
finite-dimensional problem, not uniformly over the photon cutoff.
\end{proof}

\section{Numerical implementation details}
\label{app:numerical-implementation}

The compact DM CV block form avoids constructing the Naimark-dilated pinched
state.  If $\rho=U\operatorname{diag}(p)U^\dagger$, the nonzero spectrum of
$A_a\rho A_a^\dagger$ is obtained from a thin singular-value decomposition of
\begin{equation}
  C_a\coloneqq A_aU\operatorname{diag}(\sqrt p).
  \label{eq:thin-block-factor}
\end{equation}
The same factors yield the entropy and supported-logarithm pullback.  Exact
structural zero modes remain on the declared spaces $\mathcal K_a$; a
floating rank heuristic does not redefine the channel.

For the inner Hamiltonian
\begin{equation}
  K(\lambda)\coloneqq \log\chi-\lambda\cdot M,
  \label{eq:inner-hamiltonian}
\end{equation}
one Hermitian eigendecomposition supplies the log partition, Gibbs state,
gradient, and BKM Hessian in
Eqs.~\eqref{eq:dual-gradient}--\eqref{eq:spectral-dual-hessian}.
Rejected line-search trials require eigenvalues but not eigenvectors.  No
$d^2$-by-$d^2$ superoperator is formed.  With dense observables and compact
block maps, the leading stored data scale as
\begin{equation}
  O\!\left(rd^2+r^2+d\sum_a\dim\mathcal K_a\right),
  \label{eq:storage-complexity}
\end{equation}
which becomes $O(d^2)$ when $r$ is fixed and
$\sum_a\dim\mathcal K_a=O(d)$, as in the four-block QPSK family.  Dense
spectral kernels cost $O(d^3)$ per full decomposition.

\section{Conditional symmetry reduction}
\label{app:symmetry-reduction}

Symmetries can reduce the numerical problem substantially, but only exact
invariance makes such a reduction rigorous.  This appendix proves the
unitary-twirling and antiunitary-averaging reductions under explicit assumption
and distinguishes those statements from the tolerance-based DM CV acceleration.

\subsection{Finite unitary twirling}
\label{app:unitary-twirling}

Work in the normalized DM CV specialization of
Appendix~\ref{app:channel-reduction}, with block operators
$A_a:\mathcal H\to\mathcal K_a$ and objective
$F(\rho)=S(\cZ(\rho))-S(\rho)$.  Let $\Gamma$ be a finite group and let
$g\mapsto W_g\in\operatorname{U}(\mathcal H)$ be a unitary representation.
Assume the affine slice is invariant, meaning
$W_g\rho W_g^\dagger\in\mathcal F$ for every
$\rho\in\mathcal F$ and $g\in\Gamma$.  For each $g$, assume there is a
permutation $\pi_g$ of $\mathsf A$ and, for every $a\in\mathsf A$, a unitary
$U_{g,a}:\mathcal K_a\to\mathcal K_{\pi_g(a)}$ such that
\begin{equation}
  A_{\pi_g(a)}W_g=U_{g,a}A_a
  \label{eq:block-channel-covariance}
\end{equation}
for all $g\in\Gamma$ and $a\in\mathsf A$.  Conjugation by $W_g$ then permutes
the blocks of $\cZ(\rho)$ and conjugates each by a unitary.  Hence
\begin{equation}
  F(W_g\rho W_g^\dagger)=F(\rho).
  \label{eq:unitary-objective-invariance}
\end{equation}
for every density operator $\rho$ on $\mathcal H$ and $g\in\Gamma$.
The twirl
\begin{equation}
  \mathcal E_\Gamma(\rho)
  \coloneqq \frac1{|\Gamma|}\sum_{g\in\Gamma}W_g\rho W_g^\dagger
  \label{eq:finite-group-twirl}
\end{equation}
is feasible whenever $\rho$ is feasible.  Convexity gives
\begin{equation}
  F(\mathcal E_\Gamma(\rho))
  \le\frac1{|\Gamma|}\sum_{g\in\Gamma}F(W_g\rho W_g^\dagger)
  =F(\rho).
  \label{eq:twirling-objective-bound}
\end{equation}
Applying the twirl to an optimizer proves existence of an invariant optimizer;
it does not prove that every optimizer is invariant.

\subsection{Antiunitary reflection}
\label{app:antiunitary-reflection}

Fix an orthonormal basis of $\mathcal H$, and let the overline denote entrywise
complex conjugation in that basis.  Define the following map on Hermitian
operators:
\begin{equation}
  \Theta(\rho)\coloneqq J\overline\rho J^\dagger
  \label{eq:antiunitary-involution}
\end{equation}
for a unitary $J\in\operatorname{U}(\mathcal H)$ satisfying
$J\overline J=\mathbbm1_{\mathcal H}$.  The map $\Theta$ is a real-linear,
positive, trace-preserving involution.  Suppose
$\Theta(\rho)\in\mathcal F$ and $F(\Theta(\rho))=F(\rho)$ for every
$\rho\in\mathcal F$.  Then
\begin{equation}
  \mathcal E_\Theta(\rho)\coloneqq \frac{\rho+\Theta(\rho)}2
  \label{eq:antiunitary-average}
\end{equation}
is a feasible fixed point of $\Theta$, and convexity gives
\begin{equation}
  F(\mathcal E_\Theta(\rho))\le F(\rho).
  \label{eq:antiunitary-objective-bound}
\end{equation}
Thus an invariant optimizer exists under the exact assumption.  If in addition
$\mathcal E_\Theta\mathcal E_\Gamma
=\mathcal E_\Gamma\mathcal E_\Theta$, their composition maps every optimizer
to an optimizer fixed by both averages; this is the precise compatibility
condition for the combined reduction used here.

For completeness, let $\chi$ be a positive-definite density operator on
$\mathcal H$ and suppose a finite invariant decomposition
$\mathcal H=\bigoplus_{q\in\mathsf Q}\mathcal H_q$ with orthogonal projections
$\Pi_q\in\mathcal B(\mathcal H)$ makes $\log\chi$ and every $M_i$ block
diagonal.  For
$K(\lambda)\coloneqq \log\chi-\lambda\cdot M$, define
$K_q(\lambda)\coloneqq \Pi_qK(\lambda)\Pi_q|_{\mathcal H_q}$.  The Gibbs state then
retains one global partition function,
\begin{equation}
  Z_\chi(\lambda)=\sum_q\tr\exp(K_q(\lambda)).
  \label{eq:global-sector-partition}
\end{equation}
Normalizing sectors separately changes the feasible Gibbs family and is not a
symmetry reduction of the original problem.

For the serialized Lorente QPSK instances, these relations are checked only
to declared floating tolerances.  The exact proof above establishes the
implication from exact assumption, not those assumption for the concrete
Float64 payloads.  The implemented audit uses the cyclic generator
\begin{equation}
  \begin{aligned}
    W&\coloneqq S_A\otimes\operatorname{diag}(1,\mathrm i,\ldots,\mathrm i^{N_c}),\\
    S_A\ket{j}&=\ket{j+1\bmod4},
  \end{aligned}
  \label{eq:qpsk-cyclic-generator}
\end{equation}
and the antiunitary reflection
\begin{equation}
  \begin{aligned}
    J&\coloneqq R_A\otimes\mathbbm1_{N_c+1},\\
    R_A\ket{j}&=\ket{-j\bmod4}.
  \end{aligned}
  \label{eq:qpsk-reflection}
\end{equation}
The search follows $31\to7\to6$ constraint coordinates.  These are
tolerance-admitted numerical accelerations, not exact theorems about the
serialized instances.  Every candidate is rechecked against all 31 source
moments, and every lower pencil is reconstructed in source coordinates.  The
checks diagnose lifting and reconstruction errors, but tolerance tests do not
establish exact feasibility, a directed certificate enclosure, or equality
of the reduced and unreduced infima.

\end{document}